\documentclass[copyright,creativecommons]{eptcs}

\usepackage{amsmath}
\usepackage{amssymb}
\usepackage{graphicx}
\usepackage{url}
\usepackage{stmaryrd}
\usepackage{graphicx}
\usepackage{makeidx}
\usepackage{wasysym}
\usepackage{epsfig,color,moreverb,multirow}
\usepackage{algpseudocode}
\usepackage{mathrsfs}
\usepackage{tikz}
\usetikzlibrary{arrows,automata}
\usepackage[latin1]{inputenc}
\usepackage{verbatim}
\usepackage{eurosym}
\usepackage{marginnote}

\newcommand{\qed}{\hfill $\Diamond$}

\newcommand{\liable}[1]{\mathit{Liable}({#1})}
\newcommand{\tliable}[1]{\mathit{TLiable}({#1})}

\newcommand{\ithel}[2]{{#1}_{({#2})}}
\newcommand{\blk}{\mbox{\tiny$\boxempty$}}%\usepackage{stmayrd}
\newcommand{\Oset}{\mathbb{O}}
\newcommand{\Rset}{\mathbb{R}}
\newcommand{\Lset}{\mathbb{L}}
\newcommand{\matchrel}{\bowtie}
\newcommand{\aca}[1]{\mathcal{#1}}
\newcommand{\conf}[1]{\langle{#1}\rangle}

\usepackage{ulem}
\definecolor{shadecolor}{rgb}{1,0.99,0.9}

\newtheorem{fact}{Fact}
\newtheorem{definition}{Definition}
\newtheorem{property}[fact]{Property}
\newtheorem{lemma}[fact]{Lemma}
\newtheorem{theorem}{Theorem}
\newtheorem{corollary}[fact]{Corollary}
\newtheorem{example}{Example}

\newtheorem{proposition}{Proposition}
\newtheorem{remark}{Remark}

\newenvironment{proof}{\noindent\emph{Proof.}}{}

\urldef{\mailsa}\path|basile@di.unipi.it|

\renewcommand{\epsilon}{\varepsilon}

\newcommand{\Act}{\ensuremath{\textsf{Act}}}

\newcommand{\sem}[2]{\ensuremath{\llbracket {#1} \rrbracket_{#2}}}

\newcommand{\hbra}{
\hbox to 1 \textwidth{\vrule width0.3mm height 1.8mm depth-0.3mm
                    \leaders\hrule height1.8mm depth-1.5mm\hfill
                    \vrule width0.3mm height 1.8mm depth-0.3mm}}
\newcommand{\hket}{
\hbox to 1 \textwidth{\vrule width0.3mm height1.5mm
                    \leaders\hrule height0.3mm\hfill
                    \vrule width0.3mm height1.5mm}}

\usepackage{xspace}
\usepackage{listings}
\usepackage{textcomp}

\newcommand{\keyword}[1]{\textsf{#1}\xspace}

\newcommand{\register}[1]{\keyword r_{#1}}

\newcommand{\lockbit}[1]{\textbf{#1}}
\newcommand{\participant}[1]{\mathtt{#1}}

\newcommand{\osred}{\rightarrow}                % one-step reduction relation
\newcommand{\R}{\osred}
\newcommand{\msred}{\osred^\ast}     % multi-step reduction
\newcommand{\RR}{\msred}

\newcommand{\TRANS}[1]{\xrightarrow{#1}}

\newcommand{\TRANSS}[1]{{\xrightarrow{\raisebox{-.3ex}[0pt][0pt]{\scriptsize $#1$}}}}

\newcommand{\TRANSOB}[1]{\stackrel{#1}{\twoheadrightarrow}}

\newcommand{\bottom}{\perp}
\newcommand{\PSEND}[2]{\co{ \msgsort{#2}} @ \ptp{#1}  }
\newcommand{\PRECEIVE}[2]{ \msgsort{#2} @ \ptp{#1} }

\newcommand{\lgge}[1]{\mathcal{L}(#1)}

\newcommand{\msgsort}[1]{\mathit{#1}}
\newcommand{\csconf}[2]{(\vec{#1} ;\vec{#2})}
\newcommand{\defref}[1]{Def.~\ref{#1}}

\newcommand{\emptyword}{\varepsilon}

\newcommand{\conctrans}[1]{\OPconctrans}%{\OPconctrans_{#1}}

\newcommand{\Alice}{\ptp{Alice}}
\newcommand{\Bob}{\ptp{Bob}}
\newcommand{\Carol}{\ptp{Carol}}

\DeclareMathOperator{\OPconctrans}{\blackdiamond}
\newcommand{\blackdiamond}{\diamondsuit}

\newcommand{\rcdt}{\{l_i\colon T_i\}_{i\in I}}
\newcommand{\brancht}[1][\alpha]{\branch\rcdt}

\newcommand{\AT}[2]{#1\colon\! #2}

\newcommand{\co}{\overline}

\newcommand{\ASET}[1]{\{ {#1} \}}
\newif\ifny\nytrue
\nytrue

\newif\ifvv\vvtrue
\vvfalse

\newif\ifkohei\koheitrue
\koheifalse

\newif\ifmarco\marcotrue
\marcofalse

\def\fps@figure{tp}      % Top, or separate page.
\def\fps@table{tp}
\newcommand{\node}{\mathsf{n}}

\newcommand{\ptp}[1]{\mathtt{#1}}

\newcommand{\enewchan}[2]{\ensuremath{\mathtt{newChan}\ \AT{#1}{#2}}}

\newcommand{\MRED}[1][]{%
  \ensuremath{%
    \ifthenelse{\equal{#1}{}}{%
      \rightarrow\!\!\!\!\rightarrow%
    }{%
      \rightarrow\!\!\!\!\rightarrow_{#1}%
    }%
  }%
}

\newcommand{\RED}[1][]{%
  \ensuremath{%
    \ifthenelse{\equal{#1}{}}{%
      \longrightarrow%
    }{%
      \longrightarrow_{#1}%
    }%
  }%
}

\newcommand{\LLC}%
{{\mathsf{L}\mathsf{L}\mathsf{S}}^{\mathsf{C}}}

\newcommand{\LLA}%
{{\mathsf{L}\mathsf{L}\mathsf{S}}^{\mathsf{A}}}

\newcommand{\newlinchan}[3]%%
{\ensuremath{\mathtt{newCont}\ \AT{#1}{#2}\ \mathtt{in}\ #3}}

\newcommand{\newchan}[3]%%
{\enewchan{#1}{#2}\,;\,#3}

\newcommand{\SIL}%
{{\mathsf{S}\mathsf{I}\mathsf{L}}\xspace}

\newcommand{\SILc}%
{{\mathsf{S}\mathsf{I}\mathsf{L}^{\text{\scriptsize C}}}\xspace}

\newcommand{\SILC}%
{{\mathsf{S}\mathsf{I}\mathsf{L}^{\text{\scriptsize C}}}\xspace}

\newcommand{\SILa}%
{{\mathsf{S}\mathsf{I}\mathsf{L}^{\text{\scriptsize A}}}\xspace}

\newcommand{\SILA}%
{{\mathsf{S}\mathsf{I}\mathsf{L}^{\text{\scriptsize A}}}\xspace}

{\end{array}%
 \right.}

\newcommand{\causes}[1]{\searrow}

\newif\ifdm\dmtrue
\newif\ifdmr%\dmrtrue
\dmrfalse

\newcommand{\PSet}{\!\mathcal{P}\!}

\newcommand{\ie}{i.e.~}

\newcounter{analphabet}
{\rm%
\begin{list}%
{\arabic{analphabet}. }%
{\usecounter{analphabet}%
 \addtolength{\labelwidth}{5mm}% 
\addtolength{\leftmargin}{-2mm}%
\setlength{\rightmargin}{0pt}%
\setlength{\itemsep}{0mm}%
\setlength{\parsep}{0pt}}}%
{\end{list}}

\newcommand{\mmdef}{\overset{{\text{def}}}{=}}
\newcommand{\st}{\;\big|\;} %such that set
\newcommand{\qst}{\;\colon\;} %such that logic

\newcommand{\chanset}{C}
 
\newcommand{\RS}{\mathit{RS}}

\newcommand{\p}{\ensuremath{\participant{p}}}
\newcommand{\q}{\ensuremath{\participant{q}}}

\newcommand{\rr}{\ensuremath{\participant{r}}}
\newcommand{\s}{\ensuremath{\participant{s}}}

\newcommand{\action}{\ell}
\DeclareSymbolFont{bbsymbol}{U}{bbold}{m}{n}
\DeclareMathSymbol{\bbsemicolon}{\mathbin}{bbsymbol}{"3B}

\newcommand{\ASigma}{\Rset \cup \Oset}

\newcommand{\MSACore}%
{CMSA\xspace}
\newcommand{\MSA}%
{MSA\xspace}

\providecommand{\event}{ICE 2014} % Name of the event you are submitting to
\usepackage{breakurl}             % Not needed if you use pdflatex only.

\title{From Orchestration to Choreography \\ through Contract Automata
\thanks{This work has been partially supported by the MIUR project \emph{Security Horizons}
and IST-FP7-FET open-IP project \emph{ASCENS}}}
\author{Davide Basile \quad Pierpaolo Degano \quad Gian-Luigi Ferrari 
\institute{Dipartimento di Informatica, Universit\`{a} di Pisa, Italy}
\email{\{basile,degano,giangi\}@di.unipi.it}\\
\and Emilio Tuosto
\institute{Computer Science Department, University of Leicester}
\email{emilio@le.ac.uk}
}

\def\titlerunning{From orchestration to choreography through contract automata}
\def\authorrunning{D. Basile, P. Degano, G.L. Ferrari, and E. Tuosto}
\begin{document}
\maketitle

%
%
%
%\begin{document}
%
%\mainmatter  % start of an individual contribution
%
%% first the title is needed
%
%
%% the name(s) of the author(s) follow(s) next
%%
%% NB: Chinese authors should write their first names(s) in front of
%% their surnames. This ensures that the names appear correctly in
%% the running heads and the author index.
%%

%
%%\authorrunning{}
%% (feature abused for this document to repeat the title also on left hand pages)
%
%% the affiliations are given next; don't give your e-mail address
%% unless you accept that it will be published
%
%%
%% NB: a more complex sample for affiliations and the mapping to the
%% corresponding authors can be found in the file "llncs.dem"
%% (search for the string "\mainmatter" where a contribution starts).
%% "llncs.dem" accompanies the document class "llncs.cls".
%%
%
%\toctitle{Lecture Notes in Computer Science}
%\tocauthor{Authors' Instructions}
%\maketitle

%%%%%%%%%%%%%% ABSTRACT %%%%%%%%%%%%%%

\begin{abstract}
  We study the relations between a contract automata and an
  interaction model. In the former model, distributed services are
  abstracted away as automata - oblivious of their partners - that
  coordinate with each other through an orchestrator. The interaction
  model relies on channel-based asynchronous communication and
  choreography to coordinate distributed services.

  We define a notion of strong agreement on the contract model,
  exhibit a natural mapping from the contract model to the interaction
  model, and give conditions to ensure that strong agreement
  corresponds to well-formed choreography.
\end{abstract}

%%%%%%%%%%%%%%%%%%%%%%%%%%%%%%%%%

\section{Introduction}

% !TEX root = main.tex

We investigate the relations between two models of distributed
coordination: \emph{contract automata}~\cite{BasileDF14} and
\emph{communicating machines}~\cite{BZ83}.

The former model has been recently introduced as a \emph{contract-based}
coordination framework where contracts specify the expected behaviour
of distributed components oblivious of their communicating partners.
The underlying coordination mechanism of contract automata is
orchestration.
In fact, such model envisages components capable of communicating
messages on some ports according to an automaton specifying the
component's behavioural contract.
These messages have to be thought of as directed to an orchestrator
synthesised out of the components; the orchestrator directs the
interactions in such a way that only executions that ``are in
agreement'' happen.
In this way, it is possible to transfer the approach
of~\cite{Bart2012COORDINATION,BSTZ13} to contract automata so to
identify misbehaviour of components that do not realise their
contract.

We illustrate this with the following simple example.
Alice is willing to lend her aeroplane toy, Bob offers a bike toy in
order to play with an aeroplane toy, while Carol wants to play with an
aeroplane or a bike toy.
Let $\overline{a}$ and $\overline{b}$ denote respectively the actions
of offering an aeroplane or a bike toy and, dually, $a$ and $b$ denote
the corresponding request actions.
The contract automata for Alice, Bob, and Carol correspond to the
following regular expressions, used here for conciseness:
\[
\Alice = \overline{a}
\qquad\qquad
\Bob=\overline{b}.a + a.\overline{b}
\qquad\qquad
\Carol = a + b
\]%
If Alice exchanges her toy with Bob, then all contracts are
fulfilled.
Instead, if not coordinated, Alice, Bob, and Carol may share their toys in a
way that does not fulfill their contracts. In fact, Alice can give her
aeroplane to Bob or Carol, while Carol can receive the bike from Alice
or Bob, therefore if Alice gives her aeroplane to Carol the contracts
of Alice and Carol are fulfilled while Bob's contract is not.
In the model of contract automaton the coordinator acts as the mam of
the three kids who takes their desires and suggests how to satisfy
them (and reproaches those who do not act according to their
declared contract).

Communicating machines - the other model we consider here - were
introduced with the aim of studying distributed communication
protocols and ensure the correctness of distributed components again
formalised as automata.
But - unlike contract automata - communicating machines do not
require an orchestrator since they interact directly with each other
through (FIFO) buffers.
In fact, a relation between communicating and
distributed choreographies has been recently proved in~\cite{dy12}.

We show that these models - invented to address different problems and
having different coordination mechanisms - are related.
For this purpose, we introduce the notion of \emph{strong agreement},
which requires the fulfillment of all offers \emph{and} requests.
Strong agreement differs from previous notions of agreements for
contract automata (cf. Section~\ref{sec:agreement}) and enables us to
introduce strongly safe contract automata, that is those automata
accepting only computations that are in strong agreement.
Strong agreement and safety are key to establish a correspondence from
contract automata to communicating machines.

Indeed if a contract automaton enjoys strong safety (and it is
well-behaved on branching constructs) then the corresponding
communicating machines are a 
well-formed choreography.

\paragraph{Structure of the paper.}
We recall contract automata and communicating finite-state machines in
Section~\ref{sec:blk}.
The new notion of agreement on contract automata is in Section~\ref{sec:agreement}.
The translation of contract automata into communicating machines
is given in Section~\ref{sec:mapping} where we also prove our main theorem.
In Section~\ref{sec:extensions} we discuss possible extensions of our
results  to 
other notions of agreement for contract automata and semantics for
communicating machines.
Finally, concluding remark are in Section~\ref{sec:conc}.

%
%  TODO: aggiungere nome azioni
%
\newcommand{\aout}{\overline{a}}
\newcommand{\areq}{a}
\newcommand{\cout}{\overline{c}}
\newcommand{\creq}{c}
\newcommand{\dout}{\overline{d}}
\newcommand{\dreq}{d}
\newcommand{\okout}{\overline{ok}}
\newcommand{\okreq}{ok}

%%% Local Variables: 
%%% mode: latex
%%% TeX-master: "main"
%%% End: 
\label{sec:intro}

\section{Background}\label{sec:blk}

This section summarises the automata models we use in the paper.
Both models envisage distributed computations as enacted by components
that interact by exchanging messages.
As we will see, in both cases components, abstracted
away as automata, yield systems also formalised as automata.

\subsection{Contract Automata}\label{sec:ca}
  % !TEX root = main.tex
% 
Before recalling contract automata (introduced in~\cite{BasileDF14}), we
fix our notations and preliminary definitions.
Given a set $X$, as usual, $X^\ast \mmdef \bigcup_{n \geq 0}X^n$ is
the set of finite \emph{words} on $X$ ($\emptyword$ is the empty word,
$ww'$ is the concatenation of words $w,w' \in X^\ast$, $\ithel w i$
denotes the $i$-th symbol of $w$, and $|w|$ is the length of $w$);
write $x^n$ for the word obtained by $n$ concatenations of $x \in X$
and $x^\ast$ for a finite and arbitrarily long repetition of $x \in
X$.
It will also be useful to consider $X^n$ as a set of tuples and let
$\vec x$ to range over it.
Sometimes, overloading notation (and terminology), we confound tuples
on $X$ with words on $X$ (e.g., if $\vec w \in X^n$, then $|\vec w| = n$
is the length of $w$ and $\ithel{\vec w} i$ denotes the $i$-th element of $w$).

Transitions of contract automata will be labelled with elements in the
set $\Lset \mmdef \Rset \cup \Oset \cup \{\blk\}$ where
\begin{itemize}
\item \emph{requests} of components will be built out of $\Rset$ while
  their \emph{offers} will be built out of $\Oset$,
\item  $\Rset \cap \Oset = \emptyset$, and
\item $\blk \not\in \Rset \cup \Oset$ is a distinguished label to
  represent components that stay idle.
\end{itemize}
We let $a,b,c,\ldots$ range over $\Lset$ and fix an involution
$\co \cdot : \Lset \to \Lset$ such that
\[
\co \Rset \subseteq \Oset,
\qquad\qquad
\co \Oset \subseteq \Rset,
\qquad\qquad
\forall a \in \Rset \cup \Oset \qst \co{\co a} = a,
\qquad\text{and}\qquad
\co \blk = \blk
\]

A contract automaton (cf. \defref{def:contract}) represents the
behaviour of a set of participants (possibly made of a single
participant) capable of performing some \emph{actions}; more
precisely, as formalised in \defref{def:actions}, the actions of a
contract automaton allow them to \lq\lq\ advertise\rq\rq\ offers,
\lq\lq make\rq\rq\ requests, or \lq\lq handshake\rq\rq\ on
simultaneous offer/request matches.
\begin{definition}[Actions]\label{def:actions}
  A tuple $\vec a$ on $\Lset$ is
  \begin{itemize}
  \item a \emph{request (action) on $b$} iff $\vec a$ is of the form
    $\blk^\ast b \blk^\ast$ with $b \in \Rset$
  \item an \emph{offer (action) on $b$} iff $\vec a$ is of the form
    $\blk^\ast b \blk^\ast$ with $b \in \Oset$
  \item a \emph{match (action) on $b$} iff $\vec a$ is of the form $\blk^\ast
    b \blk^\ast \co b \blk^\ast$ with $b \in \Rset \cup \Oset$.
  \end{itemize}
  We define the relation $\matchrel \subseteq \Lset^\ast \times \Lset^\ast$
  as the symmetric closure of $\stackrel \cdot \matchrel \subseteq
  \Lset^\ast \times \Lset^\ast$ where $\vec a_1 \stackrel \cdot
  \matchrel \vec a_2$ iff
  \begin{itemize}
  \item $\vec a_1$ and $\vec a_2$ are actions of the same length
  \item $\exists b \in \Rset \cup \Oset \qst \vec a_1 \text{ is an
      offer on } b \implies \vec a_2 \text{ is a request on } b$,
  \item $\exists b \in \Rset \cup \Oset \qst \vec a_1 \text{ is a
      request on } b \implies \vec a_2 \text{ is a offer on } b$,
  \end{itemize}
  We write $\vec a_1 \matchrel_b \vec a_2$ when there is $b \in \Rset
  \cup \Oset$ such that $\vec a_1$ and $\vec a_2$ are actions on $b$
  and $\vec a_1 \matchrel \vec a_2$.
\end{definition}

\begin{fact}
  $\matchrel$ is an equivalence relation on $\Lset^\ast$.
\end{fact}
% \emic{serve?}{}{Let $\co\cdot$ be extended on tuples on $\Lset$ in the obvious way}

\begin{definition}[Contract Automata]
  \label{def:contract}
  Let $\mathcal{Q}$ (ranged over by $q_1,q_2, \ldots$) be a finite set
  of states.
  A \emph{contract automaton of rank $n$} is a (finite-state)
  automaton $\aca A  = \conf{\mathcal Q ^n, \vec{q_0}, \Lset^n, T, F}$, where
  \begin{itemize}
  \item  $\vec{q_0} \in \mathcal Q^n$ is the \emph{initial} state
  %\item $R \subseteq \Rset$ and $O \subseteq \Oset$ are finite sets of labels
  \item $F \subseteq \mathcal Q^n$ is the set of \emph{accepting} states
  \item $T \subseteq \mathcal Q^n \times \Lset^n \times \mathcal Q^n$
    is the set of \emph{transitions} such that
    $(\vec q, \vec a, \vec{q}') \in T$ iff
    \begin{itemize}
    \item if $\ithel {\vec a} i = \blk$ then $\ithel {\vec q} i =
      \ithel {\vec{q}'} i$ (\ie, the $i$-th participant stays idle)
      and
    \item $\vec a$ is either a request, or an offer, or else a match action
    \end{itemize}
  \end{itemize}
  A \emph{principal} is a contract automaton $\aca A$ of rank $1$ such
  that, for any two transitions $(q_1, a_1, q'_1)$, $(q_2, a_2, q'_2)$
  in $\aca A$, it is not the case that $a_1 \matchrel a_2$.
\end{definition}

\begin{example}\label{ex:game}
The principals of Alice, Bob, and Carol in Section~\ref{sec:intro} are given below
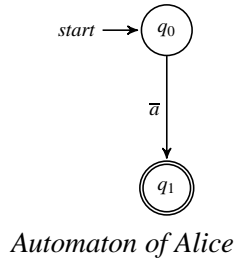
\[\begin{array}{c@{\qquad\qquad}c@{\qquad\qquad}c}
\begin{tikzpicture}[->,>=stealth',shorten >=1pt,auto,node distance=3cm,
                    semithick, every node/.style={scale=0.7}]
  \tikzstyle{every state}=[fill=white,draw=black,text=black]

  \node[initial,state] (A)                   	 {$q_0$};
  \node[state,accepting] (B)  [below of=A]{$q_1$};

  \path (A)			edge[below]             node[left]{$\overline a$} (B);
\end{tikzpicture}
&
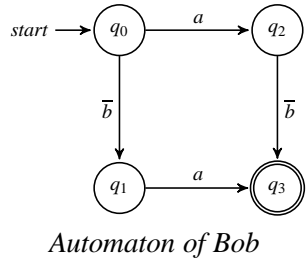
\begin{tikzpicture}[->,>=stealth',shorten >=1pt,auto,node distance=3cm,
                    semithick, every node/.style={scale=0.7}]
  \tikzstyle{every state}=[fill=white,draw=black,text=black]

  \node[initial,state] (A)                   	 {$q_0$};
  \node[state] (C) [below of = A]        	 {$q_1$};
  \node[state,accepting] (D)  [right of=C]                	 {$q_3$};
  \node[state] (B)  [right of=A]{$q_2$};

  \path (A)			edge             node[above]{$a$} (B);
  \path (B)			edge             node[right]{$\overline{b}$} (D);
  \path (C)			edge             node[above]{$a$} (D);
  \path (A)			edge             node[left]{$\overline b$} (C);
\end{tikzpicture}
&
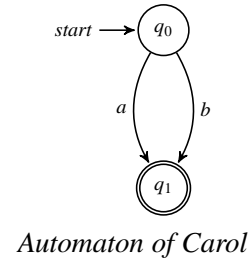
\begin{tikzpicture}[->,>=stealth',shorten >=1pt,auto,node distance=3cm,
                    semithick, every node/.style={scale=0.7}]
  \tikzstyle{every state}=[fill=white,draw=black,text=black]

  \node[initial,state] (A)                   	 {$q_0$};
  \node[state,accepting] (B)  [below of=A]{$q_1$};

  \path (A)			edge[bend left,above]             node[right]{$b$} (B);
  \path (A)			edge[bend right,below]             node[left]{$a$} (B);
\end{tikzpicture}
\\
\text{Automaton of Alice} & \text{Automaton of Bob} & \text{Automaton of Carol}
\end{array}\]
\end{example}

Given a contract automaton $\aca A = \conf{\mathcal Q^n, \vec{q_0},
  \Lset^n, T, F}$ of rank $n$, usual definitions and constructions of
finite-state automata apply.
In particular,
\begin{itemize}
\item the configurations of $\aca A$ are pairs in ${\mathcal Q}^n
  \times (\Lset^n)^\ast$ of strings of $n$-tuples of labels and states
  of $\aca A$;
\item \emph{$\aca A$ moves from $(\vec{q},w)$ to $(\vec{q}',w')$},
  written $(\vec{q},w) \TRANSS{\vec{a}} (\vec{q}',w')$, iff $w =
  \vec{a}w'$ and $(\vec{q},\vec{a},\vec{q}') \in T$; we write
  $(\vec{q},w) \to (\vec{q}',w')$ when $\vec{a}$ is immaterial and
  $\vec{q} \TRANSS{\vec{a}} \vec{q}'$ when $w$ is immaterial;
\item  the \emph{language of $\mathcal{A}$} is $\mathscr{L}(\mathcal{A})=\{w \st
  (\vec{q_0},w) \to^* (\vec{q},\epsilon), \vec{q} \in F\}$
  where $\to^*$ is the reflexive and transitive closure of
  $\to$.
  As usual, $s_1 \TRANSS{\action_1\cdots \action_m } s_{m+1}$ shortens
  $s_1 \TRANSS{\action_1} s_2 \cdots s_m \TRANSS{\action_m} s_{m+1}$
  (for some $s_2,\ldots,s_m$) and $s \not \rightarrow$ iff for no
  $s'$ it is the case that $s \rightarrow s'$.
\end{itemize}

%\emi{$\proj{}{}$ mai usata}
%\begin{definition}[Projection]
%  \label{def:proj}
%  The \emph{projection of $\aca A$ wrt $i \in \{1,\ldots, n\}$} is the
%  automaton
%  \[ \proj i {\aca A} \mmdef \conf{\mathcal Q, \ithel
%    {\vec{q_0}} i, \proj i T, \Lset, \proj i F}
%\]
%where:
%    \[
%    \proj i T = \{(\ithel {\vec q} i, \ithel {\vec a} i, \ithel{\vec{q}'} i) \st
%    (\vec{q},\vec{a},\vec{q}') \in T \ \wedge\ \ithel {\vec a} i \neq \blk\}
%    \quad\text{and}\quad
%    \proj i F = \{\ithel {\vec q} i \st \vec{q} \in F\}.
%    \]
%\end{definition}

We now borrow from~\cite{BasileDF14} the product operation of contract
automata.
Given a finite set of contract automata, this operation basically
yields the contract automaton that interleaves all their transitions
while forcing synchronisations when two contract automata are in
states ready to \lq\lq handshake\rq\rq\ (\ie, they can fire
complementary request/offer actions).
\begin{definition}[Product]
  \label{def:prod}
  Let $\aca A_i = \conf{\mathcal Q^{n_i},\vec{q_0}_i,\Lset^{n_i}, T_i, F_i}$
  be contract automata of rank $n_i$, for $i \in \{1, \ldots, h\}$.
  The \emph{product of $\aca A_1,\ldots,\aca A_h$}, denoted as
  $\bigotimes_{i \in \{1,\ldots,h \}} \mathcal A_i $, is the contract
  automaton $\conf{\mathcal Q^n, \vec{q_0}, \Lset^n, T, F}$ of rank $n
  = n_1 + \ldots + n_h$ where:

  \begin{itemize}
  \item $\vec{q_0} = \vec{q_0}_1  \ldots  \vec{q_0}_h$
  \item $F = \{\vec q_1 \ldots \vec q_h \st \forall i \in 1 \ldots h
    \qst \vec{q}_i \in F_i \}$
  \item $T$ is the least subset of $\mathcal Q^n \times \Lset^n \times
    \mathcal Q^n$ such that $(\vec{q},\vec c, \vec{q}') \in T$ iff,
    letting $\vec q = \vec q_1 \ldots \vec q_h \in \mathcal Q^n$,
    \begin{description}
    \item[either] there are $1 \leq i < j \leq h$ such that
      $(\vec{q}_i,\vec a_i,\vec{q}'_i) \in T_i$, $(\vec q_j, \vec
      a_j,\vec{q}'_j) \in T_j$, $\vec a_i \matchrel \vec a_j$ and
      \[ \left\{\begin{array}{l}
          \ithel{\vec c} i = \vec a_i, \ \ithel{\vec c} j =
          \vec a_j, \text{ and } \ithel{\vec c} l = \blk^{n_l} \text{ for } l
          \in \{1,\ldots,h\} \setminus \{i,j\} 
          \\ \text{and} \\
          \vec{q}' = \vec q_1 \ldots \vec q_{i-1} \ \vec{q}'_i\ \vec q_{i+1}
          \ldots \ \vec q_{j-1} \ \vec{q}'_j \ \vec q_{j+1} \ldots \vec q_h
        \end{array}\right.
      \]
    \item[or] $\ithel {\vec c} i = \vec a_i$, $\ithel{\vec c} l =
      \blk^{n_l}$ for each $l\neq i \in \{1,\ldots,h\}$, and $\vec{q}'
      = \vec q_1 \ldots \vec q_{i-1} \vec{q}'_i \vec q_{i+1} \ldots
      \vec q_h$ when $(\vec{q}_i,\vec a_i,\vec{q}'_i) \in T_i$ and for
      all $j \neq i$ and $(\vec q_j, \vec a_j,\vec{q}'_j) \in T_j$ it
      does not hold that $\vec a_i \matchrel \vec a_j$.
    \end{description}
  \end{itemize}
\end{definition}

\begin{example}\label{ex:gameproduct}
The contract automaton below is the product of the contract automata
in Example~\ref{ex:game}.
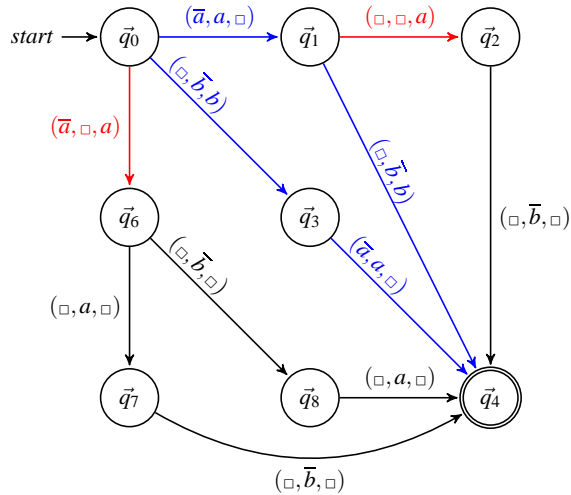
\[
\begin{tikzpicture}[->,>=stealth',shorten >=1pt,auto,node distance=3cm,
                    semithick, every node/.style={scale=0.8}]
  \tikzstyle{every state}=[fill=white,draw=black,text=black]

  \node[initial,state] (A)                   	 {$\vec{q_0}$};
  
  \node[state] (B)  [right of=A]                  			 {$\vec{q_1}$};

  \node[state] (C)  [right of=B]                 			 {$\vec{q_2}$};

\node[state] (G)  [below of=A]{$\vec{q_6}$};

 \node[state] (D)  [right of=G]      			 {$\vec{q_3}$};

\node[state] (H)  [below of=G]{$\vec{q_7}$};

\node[state] (I)  [right of=H]{$\vec{q_8}$};

\node[state,accepting] (E)  [right of=I]{$\vec{q_4}$};

  \path (A)			edge[color=blue]             node[above]{$(\overline{a},a,\blk)$} (B)
					edge[color=blue]             node{\hspace{-30pt}\rotatebox{-45}{$(\blk,\overline{b},b)$}} (D)
					edge[left,color=red]             node{$(\overline{a},\blk,a)$} (G);
  \path (B)			edge [color=red]            node{$(\blk,\blk,a)$} (C)
				edge[color=blue]             node{\hspace{-20pt}\rotatebox{-65}{$(\blk,\overline{b},b)$}} (E);
  \path (C)			edge             node{$(\blk,\overline{b},\blk)$} (E);
\path (D)			edge [color=blue]   node{\hspace{-30pt}\rotatebox{-45}{$(\overline{a},a,\blk)$}} (E);

%\path (F)			edge             node{$$} (G);

  \path (G)			edge [left]            node{$(\blk,a,\blk)$} (H)
				edge             node{\hspace{-30pt}\rotatebox{-45}{$(\blk,\overline{b},\blk)$}} (I);

\path (H)			edge[bend right,below]             node{$(\blk,\overline{b},\blk)$} (E);

\path (I)			edge             node{$(\blk,a,\blk)$} (E);
%		(A)			edge             node[below]{$(\blk,bike)$} (D)
%		(C)			edge             node[above]{$(\overline{bike},bike)$} (B)
%         (B)			edge             node[above]{$(\blk,\overline{toy})$} (A)
%		(D)			edge             node[below]{$(toy,\overline{toy})$} (E)
%		(E)			edge             node[below]{$(\overline{bike},\blk)$} (A);
\end{tikzpicture}
\]
%
%Note that the projection of the automaton above are not equal to the
%original principal but, interpreting $\blk$ as the empty string, each
%projection is equivalent to the corresponding principal.
Notice that from the states $\vec{q_0}$ and $\vec{q_6}$ (where
participants can handshake) only match actions depart; offer and
request actions are not included in the product.
\end{example}

\begin{remark}
  Notice that the product in \defref{def:prod} is not associative; an
  alternative (but more complex) definition of associative product can
  be given by ``breaking'' existing matches when composing
  automata~\cite{BasileDF14}.
\end{remark}

Hereafter, we assume that all contract automata of rank $n > 1$
are the product of $n$ principals.
Also, we consider deterministic contract automata only.  
Such assumptions could be relaxed at the cost of adding some technical
intricacies.

%%% Local Variables: 
%%% mode: latex
%%% TeX-master: "main"
%%% End: 

\subsection{Communicating Machines}\label{sec:cm}
 % !TEX root = main.tex

% Let $A,B,C, \ldots$ and $P1,P2, \ldots$ be names of
% \emph{participants} and $a,b,c,\ldots \in \Sigma^r$.

%\color{red}

Communicating machines~\cite{BZ83} are a simple automata-model introduced to
specify and analyse systems made of agents interacting via
asynchronous message passing.
We adapt the original definitions and notation from~\cite{BZ83}
and~\cite{cf05} to our needs; in particular, the only relevant
difference with the original model is that we have to add the set of
final states.
Let $\PSet$ be a finite set of \emph{participants} (ranged over by
$\p$, $\q$, $\rr$, $\ptp s$, etc.)  and $\chanset \mmdef \ASET{\p\q \st
  \p,\q \in \PSet \text{ and }\p \neq \q}$ be the set of
\emph{channels}.

\begin{remark}
  The set $\PSet$ can be thought of as the set of integers $\{1,
  \ldots, n \}$ (and likewise for contract automata). However, we
  adopt a different notation to make the translation from contract
  automata to communicating machines clearer.
\end{remark}

The set of \emph{actions} is $\Act \mmdef \chanset \times (\ASigma)$
and it is ranged over by $\action$; we abbreviate $(\ptp{sr},\co{a})$
with $\PSEND{sr}{a}$ when $\co a \in \Oset$ (representing the
\emph{sending} of $a$ from machine $\s$ to $\rr$) and, similarly, we
shorten $(\ptp{sr},a)$ with $\PRECEIVE{sr}{a}$ when $a \in \Rset$ (representing the
\emph{reception} of $a$ by $\rr$).
\begin{definition}[CFSM]\label{def:cfsm2} 
  A \emph{communicating finite-state machine} is an automaton $M =
  (Q,q_0, \ASigma,\delta, F)$ where $Q$ is a finite set of {\em
    states}, $q_0\in Q$ is the \emph{initial} state, $\delta\
  \subseteq \ Q \times \Act \times Q$ is a set of \emph{transitions},
  and $F \subseteq Q$ is the set of final states.
  We say that $M$ is \emph{deterministic} iff for all states $q \in Q$
  and all actions $\action \in \Act$, if $(q,\action,q'),
  (q,\action,q'')\in \delta$ then $q' = q''$.
  Also, we write $\lgge{M} \subseteq \Act^{\ast}$ for the language on
  $\Act$ accepted by the automaton corresponding to machine $M$.
\end{definition}
We will consider only deterministic CFSMs. The notion of deterministic CFSMs
adopted here  differs from the
  standard one which requires that, for any
  state $q$, if $(q,\PSEND{sr}{a},q') \in \delta$ and
  $(q,\PSEND{sr}{b},q'')\in \delta$ then $a=b$ and $q'=q''$
  (see e.g.,~\cite{cf05}).
  The reason for the definition is to reflect the semantics of
  contract automata.

The communication model of CFSMs (cf. Definitions~\ref{def:cs}
and~\ref{def:rs}) is based on (unbounded) FIFO buffers - the channels
in $\chanset$ - used by participants to exchange messages.
To spare another syntactic category and cumbersome definitions, we
draw the messages appearing in the buffers of CFSMs from the set
of requests $\Rset$.
Recall that the set of participants $\PSet$ is finite.
\begin{definition}[Communicating systems]\label{def:cs}
  Given a CFSM $M_\p=(Q_\p, q_{0\p},\ASigma,\delta_\p,F_\p)$ for
  each $\p \in \PSet$, the tuple $S=(M_\p)_{\p\in\PSet}$ is a
  \emph{communicating system} (\emph{CS}).
  A \emph{configuration} of $S$ is a pair $s = \csconf q w$ where
  $\vec q = (q_\p)_{\p\in\PSet}$ with $q_\p\in Q_\p$ and where
  $\vec{w}=(w_{\p\q})_{\p\q\in\chanset}$ with $w_{\p\q}\in
  \Rset^\ast$; component $\vec q$ is the \emph{control state} and
  $q_\p\in Q_\p$ is the \emph{local state} of machine $M_\p$. 
  The \emph{initial configuration} of $S$ is $s_0 =
  \csconf{q_0}{\emptyword}$ with $\vec{q_0} = (q_{0\p})_{\p\in\PSet}$.
\end{definition}
Hereafter, we fix a machine $M_\p=(Q_\p,
q_{0\p},\ASigma,\delta_\p,F_\p)$ for each participant $\p \in \PSet$
and let $S=(M_\p)_{\p\in\PSet}$ be the corresponding system.

% Communicating machines induce the following relation on the configurations
% of $S$.

% \begin{example}
%   \NY{We move this example to Introduction}
%   We will use the following communicating system as a running example
%   in this paper.
%   \begin{center}
%     \includegraphics[scale=0.35]{../misc/cavrunningex_machines.png}
%   \end{center}
% \end{example}
% 

\begin{definition}[Reachable state]\label{def:rs}
  % Let $S$ be a communicating system. 
  A configuration $s'=(\vec{q}';\vec{w}')$ is {\em reachable} from
  another configuration $s=(\vec{q};\vec{w})$ by {\em firing
    $\action$}, written $s \TRANSS{\action} s'$,
  if there is $a \in \Rset$ such that
  \begin{enumerate}
  \item either $\action= \PSEND{sr}{a}$ and $(\ithel {\vec q} \s,
    \action, \ithel {\vec q'} \s) \in \delta_{\s}$, $\ithel {\vec q'}
    \p = \ithel {\vec q} \p$ for all $\p \neq \s$, and $\ithel {\vec
      w'} {\s \rr} = \ithel {\vec w} {\s \rr} \cdot a$ and, for all
    ${\p\q} \neq \s\rr$, $\ithel {\vec w'}{\p\q} = \ithel {\vec w}
    {\p\q}$
  \item or $\action= \PRECEIVE{sr}{a}$ and $(\ithel {\vec q} \rr,
    \action, \ithel {\vec q'} \rr)\in \delta_\rr$, $\ithel {\vec q'}
    \p = \ithel {\vec q} \p$ for all $\p \neq \rr$, and $\ithel {\vec
      w'} {\s \rr} = a \cdot \ithel {\vec w} {\s \rr}$ and $\ithel {\vec
      w'} {\p\q} = \ithel {\vec w} {\p\q}$ for all ${\p\q} \neq
    \ptp{sr}$.
  \end{enumerate}
  We write $s \TRANS{} s'$ for $\exists \action \qst s
  \TRANSS{\action} s'$ and denote with $\RR$ the reflexive and
  transitive closure of $\R$.
The set of \emph{reachable configurations of $S$} is $\RS(S)=\ASET{
  s \st s_0 \RR s}$.
A sequence of transitions is
\emph{$k$-bounded} if no channel of any intermediate configuration
on the sequence contains more than $k$ messages. 
% 
% The \emph{$k$-reachability set of $S$} is the largest subset $\RS_k(S)$ of $\RS(S)$ within which each configuration $s$ can be reached by a $k$-bounded execution from $s_0$.
\end{definition}
Condition (1) in \defref{def:rs} puts the content $a$ on a channel
$\ptp{sr}$, while (2) gets the content $a$ from $\ptp{sr}$.

%%% Local Variables: 
%%% mode: latex
%%% TeX-master: "main"
%%% End: 

\subsection{Notational Synopsis}
 % !TEX root = main.tex

To avoid their continuous repetition, through the paper we assume
fixed a contract automaton $\aca A = \conf{\mathcal Q^n, \vec{q_0},
  \Lset^n, T, F}$ of rank $n$.

For readability we summarise the notations introduced so far in the
following table.
  \[
  \begin{array}{l@{\quad}l}
    X^\ast & \text{set of finite words on a set $X$; $\emptyword$ is the empty word}
    \\
    \ithel w i & \text{the $i$-th symbol of $w$}
    \\
    |w| & \text{ the length of $w$}
    \\
    x^n \text{ (resp. $x^\ast$)} & \text{$x$ concatenated $n$-times (resp. arbitrarily many) with itself}
    \\
    \vec x \text{ or } (x_i)_{1 \leq i \leq n} & \text{ indexed tuples}
    \\
    \Lset & \text{labels (ranged over by $a$, $b$, $c$, etc.)}
    \\
    \Rset & \text{request labels}
    \\
    \Oset & \text{offer labels}
    \\
    \blk \not\in \Rset \cup \Oset & \text{idle label}
    \\
    \aca A & \text{contract automata of rank $n$}
    \\
    \PSet & \text{set of participants (ranged over by
$\p$, $\q$, $\ptp i$, $\ptp j$,$\ptp A$, $\ptp B$,$\ptp C$, etc.)}
    \\
    \chanset & \text{set of channels (ranged over by $\p\q $)}
    \\
    M_\p & \text{communicating machine of participant $\p$}
    \\
    S & \text{a system of communicating machines}
  \end{array}
  \]

Finally, we assume that the states of any automaton/machine are build
out of a universe $\mathcal Q$ (of states).

%%% Local Variables: 
%%% mode: latex
%%% TeX-master: "main"
%%% End: 

\section{Enforcing Agreement}\label{sec:agreement}
  % !TEX root = main.tex
This section introduces a new notion of agreement on contract automata
- called \emph{strong agreement} - that elaborates the notions of
\emph{agreement} and \emph{weak agreement} introduced
in~\cite{BasileDF14}.
The three notions differ on the conditions for the fulfillment of an
interaction between different principals.
Briefly, an \emph{agreement} exists if all the requests, but not necessarily all the offers, are satisfied
synchronously. Intuitively, this means that the orchestrator ``simultaneously''
guarantees two participants that their complementary actions are matched. 
Instead, a \emph{weak agreement} exists when request actions can
be performed \lq\lq on credit\rq\rq.
In other words, a computation yelds weak agreement when the fulfillment
of a request action can happen after the action has been taken.
Intuitively, this corresponds to an asynchronous communication
admitting actions taken on credit provided that obligations will be
honored later on.

Here, we focus on \emph{strong agreement}, which strengthens the
previous notion of agreement by requiring the fulfillment of all offers
\emph{and} requests in a synchronous way.
In Section \ref{sec:mapping} we will show how this condition corresponds to 
interactions between communicating machines.

\begin{definition}[Strong Agreement and Safety]\label{def:safety}
  A \emph{strong agreement on $\Lset$} is a finite (non-empty)
  sequence of match actions.
  We let $\mathfrak{Z}$ to denote the set of all strong agreements on $\Lset$.

  A contract automaton ${\aca A}$ is \emph{strongly safe} if
  $\mathscr{L}({\aca A})\subseteq \mathfrak{Z}$, otherwise it is
  \emph{strongly unsafe}.
  We say that ${\aca A}$ \emph{admits strong agreement} when
  $\mathscr{L}(\aca A) \cap \mathfrak{Z} \neq \emptyset$.
\end{definition}
Note that $\emptyword$ does not belong to $\mathfrak Z$; the reason is
that $\emptyword$ would not be an interesting agreement because it
does not require any interaction (neither with the controller nor
between principals).
For this we require that the initial states of contract automata are
not accepting states.

We show how to generate a strongly safe composition of contracts with
an approach borrowed by the supervisory control theory for discrete
event systems~\cite{Cassandras2006}.
In this  theory, discrete event systems are basically automata where
accepting states represent the successful termination of a few tasks
while \emph{forbidden states} are those that should not be traversed
in ``good'' computations.
The purpose is then to synthesize a controller that enforces this
property.
The supervisory control theory distinguishes between
\emph{controllable} events (those events that the controller can
disable) and \emph{uncontrollable} events (those that are always
enabled).
Moreover, the theory partitions events in \emph{observable} and
\emph{unobservable}; the latter being a subset of uncontrollable
events.
It is known that if all events are observable then a maximally
permissive controller exists that never blocks a good computation~\cite{Cassandras2006}.

Since the behaviors that we want to enforce in ${\aca A}$ are exactly
those traces labeled by words in \linebreak $\mathfrak{Z} \cap \mathscr{L}({\aca
  A})$, we specialise the notions of supervisory control theory by
defining
\begin{itemize}
\item observable events to be all offer, request, and
  match actions;
\item forbidden events to be non-match actions.
\end{itemize}

\begin{definition}[Controller]\label{def:controller}
  A \emph{(strong) controller} of $\aca A$ is a contract automaton
  $KS_{{\aca A}}$ such that \linebreak $\mathscr{L}(KS_{{\aca A}}) \subseteq
  \mathfrak{Z} \cap \mathscr{L}(\aca A)$.
  The \emph{most permissive (strong) controller } (MPC) of $\aca A$ is the
  controller $KS_{{\aca A}}$ such that
  $\mathscr{L}(KS'_{{\aca A}}) \subseteq
  \mathscr{L}(KS_{{\aca A}})$ for all $KS'_{\aca A}$ controllers of
  $\aca A$.
\end{definition}
Note that the most permissive controller is unique up-to language equivalence.  

\begin{example}\label{ex:gameMPC}
  The MPC of the contract automaton in Example~\ref{ex:gameproduct} consists
of the states $\vec q_0$,  $\vec q_1$,  $\vec q_3$, and  $\vec q_4$
with transitions 
$(\vec q_0,(\overline a, a, \blk ) , \vec q_1)$,
$(\vec q_3,(\overline a, a, \blk ) , \vec q_4)$,
$(\vec q_0,(\blk, \overline b, b) , \vec q_3)$, and
$(\vec q_1,(\blk, \overline b, b) , \vec q_4)$.
\end{example}
\begin{proposition}
  If $KS_{{\aca A}}$ is the most permissive controller of ${\aca A}$ then
  $\mathscr{L}(KS_{{\aca A}})=\mathfrak{Z} \cap
  \mathscr{L}({\aca A})$.
\end{proposition}
\begin{proof}
  By contradiction, assume $\mathscr{L}(KS_{{\aca A}})\subset
  \mathfrak{Z} \cap \mathscr{L}({\aca A})$.
  Since $\mathfrak{Z} \cap \mathscr{L}({\aca A})$ is the
  intersection of two regular languages and all actions are
  controllable, there exists a contract
  automaton $KS'_{{\aca A}}$ accepting it (cf. \cite{Cassandras2006}).
  By definition, $KS'_{{\aca A}}$ is a controller of $\aca A$ strictly
  containing 
  $\mathscr{L}(KS_{{\aca A}})$, contradicting the hypothesis that
  $\mathscr{L}(KS_{{\aca A}})$ is the  most permissive controller.
  \qed
\end{proof}

A state $\vec q$ of a contract automaton $\aca A$ is called
\emph{redundant} if, and only if, from $\vec q$ no accepting
state of $\aca A$ can be reached.

\begin{lemma}[MPC]
\label{lem:controller}
A contract automaton is the most permissive controller of $\aca
A$ if, and only if, it is language-equivalent to
\[
KS_{\aca A} \mmdef \conf{\mathcal Q^n, \vec{q_0}, \Lset^n,
  T' \setminus \{ (\vec{q},a,\vec{q}') \st \vec{q} \text{ or }
  \vec{q}' \text{ is redundant in } \mathcal K\}, F}
\]
where $\mathcal K = \conf{\mathcal Q^n, \vec{q_0}, \Lset^n,\{t \in T \st t \text{ is a match transition}\}, F}$
is the sub-automaton of $\aca A$ consisting
of the match transitions of $\aca A$ only.
\end{lemma}
\begin{proof}
  By construction, the transitions of $KS_{{\aca A}}$ are a subset
  of the transitions of $\aca A$, hence $\mathscr{L}(KS_{\aca
    A}) \subseteq \mathfrak{Z} \cap \mathscr{L}({\aca A})$.
  Therefore  $KS_{{\aca A}}$ is a controller of $\aca A$ and we have just
  to prove that $\mathscr{L}(KS_{{\aca A}}) = \mathfrak{Z} \cap
  \mathscr{L}({\aca A})$.
  We proceed by contradiction.

  Let $w \in (\mathfrak{Z} \cap \mathscr{L}({\aca A}) )\setminus
 \mathscr{L}(KS_{\aca A})$.
  Since $\mathfrak Z$ does not contain the empty string we have $w \neq \emptyword$ and there must be a
  transition $t = (\vec q, \vec a, \vec q')$ not in $KS_{\aca A}$ in
  the accepting path of $w$ (which is unique since we consider
  deterministic contract automata only), otherwise $w \in  \mathscr{L}(KS_{\aca A})$. We know that  $\vec a$ is a match action because $w \in \mathfrak{Z}$,  and $\vec{q}$,$\vec{q}'$ are not redundant states of $\aca A$ because the transition belongs to an accepting path. Hence there must be $t \in KS_{\aca A}$, since  by construction  match transitions between non-redundant states are in $KS_{\aca A}$.
  \qed
\end{proof}

\begin{example}
The MPC of Example \ref{ex:gameMPC} is obtained from the
CA in Example \ref{ex:gameproduct} by applying the 
construction of Lemma \ref{lem:controller}.
\end{example}

The \emph{controlled system} of a contract automaton $\aca A$
identifies the match transitions of $\aca A$ and those transitions
that lead ``outside'' of the controller; for this we use a
distinguished state $\bot \not \in \mathcal Q^n$ (for any $n$) in the
following definition.
\begin{definition}[Controlled system]\label{def:controlled}
  Let $KS_{\aca A} = \conf{\mathcal Q^n, \vec{q_0},
    \Lset^n, T' \subseteq T, F}$ be the MPC of
  ${\aca A}$ as computed in Lemma~\ref{lem:controller}.
  The \emph{controlled system of $\aca A$ under $KS_{\aca A}$} is
  defined as the automaton \linebreak $KS_{{\aca A}}/{\aca A} = \conf{\mathcal
    Q^n \cup \{\bot\}, \vec{q_0}, \Lset^n, T'', F}$ such that
  \[
  T'' \quad = \quad T' \quad \cup \quad \{ (\vec q, \vec a, \bot) \st \vec q
  \text{ reachable from $\vec{q_0}$ in $KS_{\aca A}$ and }
  \exists \vec{q}' \in \mathcal Q^n \qst (\vec{q},\vec{a},\vec{q}') \in T \setminus
  T'\}\]
\end{definition}
\begin{example}
The controlled system of the CA in Example \ref{ex:gameproduct} is obtained
by adding the transitions \linebreak
$(\vec q_1,(\blk,\blk,a),\vec q_2),
(\vec q_0, (\overline a, \blk, a), \vec q_6)$
to the MPC of Example \ref{ex:gameMPC}.
\end{example}
It is worth remarking that the transitions reaching $\bot$ in the
controlled system of $\aca A$ identify the start of the computations
in $\aca A$ which lead to violations of strong agreement.

In the next definition, we introduce a notion of strong liability,
to single out the principals that are potentially responsible of the
divergence from the expected behaviour.
\begin{definition}[Strong Liability]
\label{def:culpability}
Given a controlled system $KS_{{\aca A}}/{\aca A}$,
the set of \emph{liable} participants on a trace 
$w \in \mathscr{L}(\aca A)$  is given by:
\[
\liable{KS_{{\aca A}}/{\aca A},w} = \{1 \leq \ptp i \leq n \mid (\vec{q_0}, w) \to^* (\vec
q,\vec{a}w') \to (\bot,w) \ in \ KS_{{\aca A}}/{\aca A}, \ithel {\vec a}{\ptp  i} \neq \blk \}
\]
%
% if $v\vec{a}w \in
%\mathscr{L}({\aca A})$ and $(\vec{q_0}, v\vec{a}w) \to^* (\vec
%q,\vec{a}w) \to (\bot,w)$ is a path in $KS_{\aca A}/{\aca A}$ then (a
%principal with) index $1 \leq \ptp i \leq n$ is \emph{liable on
%  $v\vec{a}w$} (stated as $\ptp i \in \liable{KS_{{\aca A}}/{\aca A},v\vec{a}w}$)
%if, and only if, $\ithel {\vec a}{\ptp  i} \neq \blk$.
%
The \emph{potentially liable principals} in $KS_{{\aca A}}/{\aca
  A}$ are $\liable{KS_{{\aca A}}/{\aca A}} \mmdef \bigcup_{w \in
  \mathcal{L}(\aca A)}\liable{KS_{{\aca A}}/{\aca A},w}$.
  
  We let $\tliable{KS_{{\aca A}}/{\aca A}}$ to denote the set of
transitions of $\aca A$ that make principals liable.
\end{definition}
Note that the transition labelled
by $\vec a$ in Definition~\ref{def:culpability} is the first which diverges from the expected path
(since, by Definition~\ref{def:controlled}, state $\bot$ does not have
outgoing transitions).
Indeed a liable index identifies a principal that fires an action
taking the computation away from agreements.

\begin{example}\label{ex:gameliable}
  The liable indexes of the contract automaton in
  Example~\ref{ex:gameproduct} are $1$ and $3$, corresponding to Alice
  and Carol respectively; the transitions that make them liable are
  respectively $\big(\vec q_0 , \ (\overline a, \blk, a),\ \vec
  q_6\big)$ and $\big(\vec q_1, \ (\blk, \blk, a), \ \vec q_2\big)$.
  The former liable transition is a match  that
  leads to a non-match transition.
\end{example}

Note that labels allow us to track participants firing
actions so to find (the indexes of) the liable principals.
Our aim is to restrict the behaviour of principals so that they follow
only the traces of the automaton which lead to strong agreement, while
avoiding the others.

%%% Local Variables: 
%%% mode: latex
%%% TeX-master: "main"
%%% End: 

\section{From Contract Automata to Communicating Machines}\label{sec:mapping}
  % !TEX root = main.tex
% 
The translation of a principal into a communicating machine is
conceptually straightforward.
Indeed, the translation just yields a machine isomorphic to a principal in the composed contract automaton; 
the only difference are the labels.
To account for the ``openness'' nature of contract automata - where
principals can fire transitions not matched by other principals - in
Definition~\ref{def:tr} below we use the '$-$' symbol representing a
special (``anonymous'') participant distinguished by the participants
corresponding to the principals and playing the role of the
environment.
For this reason, we will assume from now on that actions in $\Act$ are
built on $\chanset \cup \{-\}$.

\begin{definition}[Translation]\label{def:tr}
  The translation $\sem{\vec{a}}{\ptp p} \in \Act$ of an action $\vec{a}$ on $\Lset^n$
  respect to a participant $\ptp p$ (with $1 \leq \ptp p \leq n$) is defined as:
  \[
  \sem{\vec{a}}{\ptp p} = \left\{ 
    \begin{array}{ll} 
      \PSEND{ij}{a}& \text{if $\vec a$ is a match action and $i$
        and $j$ are such that $\ithel{\vec{a}}{ i} \in \Oset$ and
        $\ithel{\vec{a}}{j} \in \Rset$ and $\ptp p=\ptp i$}
      \\
      \PRECEIVE{ij}{a}& \text{if $\vec a$ is a match action and
        $i$ and $j$ are such that $\ithel{\vec{a}}{ i} \in \Oset$
        and $\ithel{\vec{a}}{j} \in \Rset$ and $\ptp p=\ptp j$}
      \\
      \PSEND{i-}{a}& \text{if $\vec a$ is an offer action and
        $i$ is such that $\ithel{\vec{a}}{ i} \in \Oset$ and $\ptp
        p=\ptp i$} \\ \PRECEIVE{-j}{a}& \text{if $\vec a$ is a
        request action and $j$ is such that $\ithel{\vec{a}}{j} \in
        \Rset$ and $\ptp p=\ptp j$}
	\\
	\epsilon & \text{otherwise}
    \end{array} 
  \right.
  \]

  The \emph{translation of $\aca A$ to a CFSM} is given
  by the map
  \[
  \sem{\aca A}{\ptp p} \mmdef
  \conf{\mathcal Q,
    \ithel {\vec{q_0}} p, \Act,
    \{(\ithel {\vec q} p, \sem{\vec a}{\ptp p},
    \ithel {\vec q'} p) \st (\vec
    q, \vec a, \vec q') \in T \text{ and } \sem{\vec a}{\ptp p} \neq
    \emptyword\}, F}
  \]

  We denote with $S(\aca A) = (\sem{\aca A}{\ptp p})_{\ptp p
    \in \{1, \ldots, n \}}$ the communicating system obtained by translating the
  contract automaton $\aca A$.
\end{definition}

Given $\varphi \in \big(\Lset^n\big)^\ast$, 
%\textcolor{red}{   \marginpar{non serve piu'}we say that $\varphi$ is
%\emph{$\ptp p$-free} if $\ptp p$ does not occur in the actions in
%$\varphi$.}
% 
we define
\[
\llbracket \varphi \rrbracket \mmdef \left\{ 
  \begin{array}{ll} 
    \PSEND{ij}{a} \ \PRECEIVE{ij}{a}\llbracket \varphi' \rrbracket
    & \text{if $\varphi=\vec{a}\varphi'$ and $\vec a$ is a match action
      on $a$ with $\ithel{\vec{a}}{ i} \in \Oset$ and $\ithel{\vec{a}}{j} \in \Rset$}
    \\
    \PSEND{i-}{a}\llbracket \varphi' \rrbracket &
    \text{if $\varphi=\vec{a}\varphi'$ and $\vec a$ is an offer action
      on $a$ with $\ithel{\vec{a}}{ i} \in \Oset$ }
    \\
    \PRECEIVE{-j}{a}\llbracket \varphi' \rrbracket &
    \text{if $\varphi=\vec{a}\varphi'$ and $\vec a$ is a request action
      on $a$ with $\ithel{\vec{a}}{j} \in \Rset$}
    \\
    \emptyword & \text{if $\varphi=\emptyword$}
    \\
    \text{undefined} & \text{otherwise}
  \end{array} 
\right. 
\]
\begin{example}\label{ex:translation}
Consider the following principal CAs:
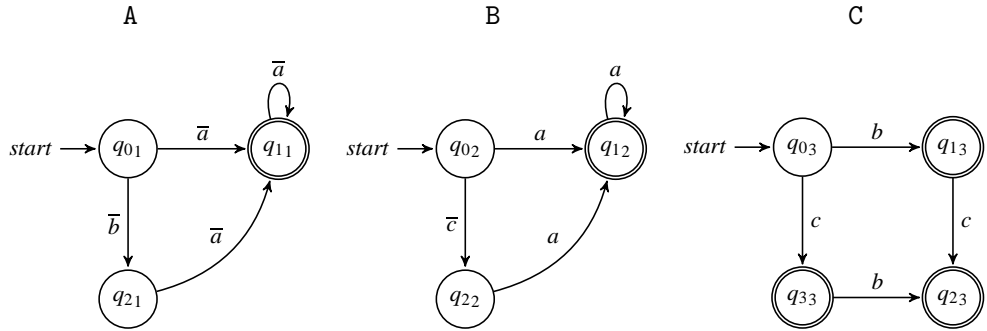
\[
\ptp A   \qquad \qquad \qquad \qquad \qquad \qquad  
\ptp B    \qquad \qquad \qquad \qquad \qquad \qquad  
\ptp C
\]
\[
\begin{tikzpicture}[->,>=stealth',shorten >=1pt,auto,node distance=2.5cm,
                    semithick, every node/.style={scale=0.8}]
  \tikzstyle{every state}=[fill=white,draw=black,text=black]

  \node[initial,state] (A)                   	 {${q_0}_1$};
  
  \node[state, accepting] (B)  [right of=A]              	{${q_1}_1$};
  
  \node[state] (C)	[below of=A]         {${q_2}_1$};

  \path (A)			edge              node{$\overline{a}$} (B)
					edge              node[left]{$\overline b$} (C)
        (C)			edge	[bend right]			  node{$\overline{a}$}(B)
        	(B)			edge[loop above]	  node{$\overline{a}$}(B);				
					
\end{tikzpicture} \quad
\begin{tikzpicture}[->,>=stealth',shorten >=1pt,auto,node distance=2.5cm,
                    semithick, every node/.style={scale=0.8}]
  \tikzstyle{every state}=[fill=white,draw=black,text=black]

  \node[initial,state] (A)                   	 {${q_0}_2$};
  
  \node[state, accepting] (B)  [right of=A]              	{${q_1}_2$};
  
  \node[state] (C)	[below of=A]         {${q_2}_2$};

  \path (A)			edge              node{$a$} (B)
					edge              node[left]{$\overline c$} (C)
        (C)			edge	 [bend right]			  node{$a$}(B)
        	(B)			edge[loop above]	  node{$a$}(B);				
					
\end{tikzpicture} \quad
\begin{tikzpicture}[->,>=stealth',shorten >=1pt,auto,node distance=2.5cm,
                    semithick, every node/.style={scale=0.8}]
  \tikzstyle{every state}=[fill=white,draw=black,text=black]

  \node[initial,state] (A)                   	 {${q_0}_3$};
  
  \node[state,accepting] (B)  [right of=A]        {${q_1}_3$};
  
  \node[state,accepting] (C)	[below of=B]         {${q_2}_3$};
   
  \node[state,accepting] (D)	[below of=A]         {${q_3}_3$};

  \path (A)			edge              node{$b$} (B)
  					edge				  node{$c$} (D)
        (B)			edge				  node{$c$}(C)
        	(D)			edge				  node{$b$}(C);				
					
\end{tikzpicture}
\]
\\
\\
\\
The product is the contract automaton below with initial state $\vec q_0 = \langle {q_0}_1,{q_0}_2,{q_0}_3 \rangle$
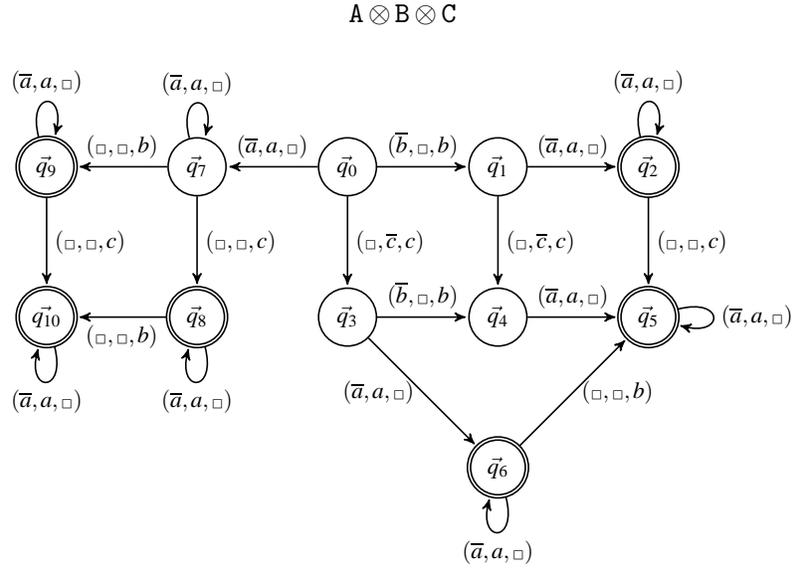
\[
\ptp A \otimes \ptp B \otimes \ptp C
\]
\[
\begin{tikzpicture}[->,>=stealth',shorten >=1pt,auto,node distance=2.5cm,
                    semithick, every node/.style={scale=0.8}]
  \tikzstyle{every state}=[fill=white,draw=black,text=black]

  \node[state,accepting] (L)                   	 {$\vec{q_{9}}$};
  \node[state] (H) [right of=L]                  	 {$\vec{q_7}$};
  \node[state] (A) [right of=H]                   	 {$\vec{q_0}$};
  \node[state] (B) [right of=A]                  	 {$\vec{q_1}$};
  \node[state,accepting] (C)  [right of=B]                 	 {$\vec{q_2}$};
  \node[state,accepting] (M) [below of=L]                  	 {$\vec{q_{10}}$};
  \node[state,accepting] (I) [right of=M]                  	 {$\vec{q_8}$};
  \node[state] (D)        [right of=I]           	 {$\vec{q_3}$};
  \node[state] (E)   [right of=D]                	 {$\vec{q_4}$};
  \node[state,accepting] (F)  [right of=E]                 	 {$\vec{q_5}$};
  \node[state,accepting] (G)  [below of=E]                 	 {$\vec{q_6}$};		
  
  \path (A)			edge              node[above]{$(\overline a, a, \blk)$} (H)
  					edge              node{$(\blk,\overline c, c)$} (D)
  					edge              node{$(\overline b,\blk,b)$} (B);
  \path (B)			edge              node{$(\overline a, a, \blk)$} (C)
  					edge              node{$(\blk,\overline c,c)$} (E);
  \path (C)			edge [loop above]         node{$(\overline a, a, \blk)$} (C)
  					edge              node{$(\blk , \blk, c)$} (F);
  \path (D)			edge              node{$(\overline b, \blk, b)$} (E)
  					edge              node[left]{$(\overline a, a, \blk)$} (G);
  \path (E)			edge              node{$(\overline a, a, \blk)$} (F);
  \path (F)			edge [loop right]             node{$(\overline a, a, \blk)$} (F);
  \path (G)			edge              node[right]{$(\blk,\blk,b)$} (F)
  					edge [loop below]             node{$(\overline a, a, \blk)$} (G);
  \path (H)			edge [loop above]             node{$(\overline a, a, \blk)$} (H)
  					edge              node[above]{$(\blk, \blk, b)$} (L)
  					edge              node{$(\blk , \blk, c)$} (I);
  \path (I)			edge  [loop below]            node{$(\overline a, a, \blk)$} (H)
  					edge 	          node{$(\blk, \blk, b)$} (M);
  \path (L)			edge  [loop above] node{$(\overline a, a, \blk)$} (L)
  					edge              node{$(\blk, \blk ,c)$} (M);
  \path (M)			edge [loop below]             node{$(\overline a, a, \blk)$} (M);				
\end{tikzpicture}
\]
By applying Lemma \ref{lem:controller} on the product automaton we obtain the MPC:
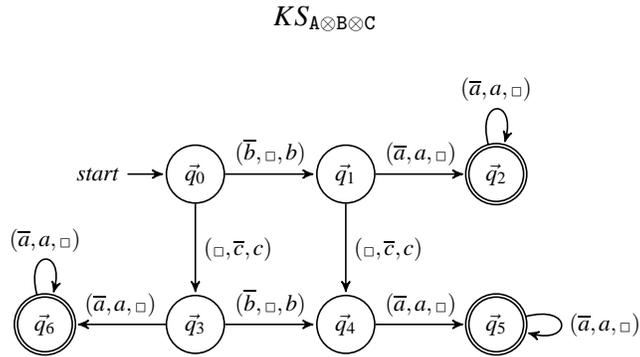
\[
KS_{\ptp A \otimes \ptp B \otimes \ptp C}
\]
\[
\begin{tikzpicture}[->,>=stealth',shorten >=1pt,auto,node distance=2.5cm,
                    semithick, every node/.style={scale=0.8}]
  \tikzstyle{every state}=[fill=white,draw=black,text=black]

  \node[state, initial] (A) [right of=H]                   	 {$\vec{q_0}$};
  \node[state] (B) [right of=A]                  	 {$\vec{q_1}$};
  \node[state,accepting] (C)  [right of=B]                 	 {$\vec{q_2}$};
  \node[state] (D)        [right of=I]           	 {$\vec{q_3}$};
  \node[state] (E)   [right of=D]                	 {$\vec{q_4}$};
  \node[state,accepting] (F)  [right of=E]                 	 {$\vec{q_5}$};
  \node[state,accepting] (G)  [left of=D]                 	 {$\vec{q_6}$};		
  
  \path (A)			
  					edge              node{$(\blk,\overline c, c)$} (D)
  					edge              node{$(\overline b,\blk,b)$} (B);
  \path (B)			edge              node{$(\overline a, a, \blk)$} (C)
  					edge              node{$(\blk,\overline c,c)$} (E);
  \path (C)			edge [loop above]         node{$(\overline a, a, \blk)$} (C); 
  \path (D)			edge              node{$(\overline b, \blk, b)$} (E)
  					edge              node[above]{$(\overline a, a, \blk)$} (G);
  \path (E)			edge              node{$(\overline a, a, \blk)$} (F);
  \path (F)			edge [loop right]             node{$(\overline a, a, \blk)$} (F);
  \path (G)			
  					edge [loop above]             node{$(\overline a, a, \blk)$} (G);				
\end{tikzpicture}
\]

The translation of Definition \ref{def:tr} yelds the CMs:

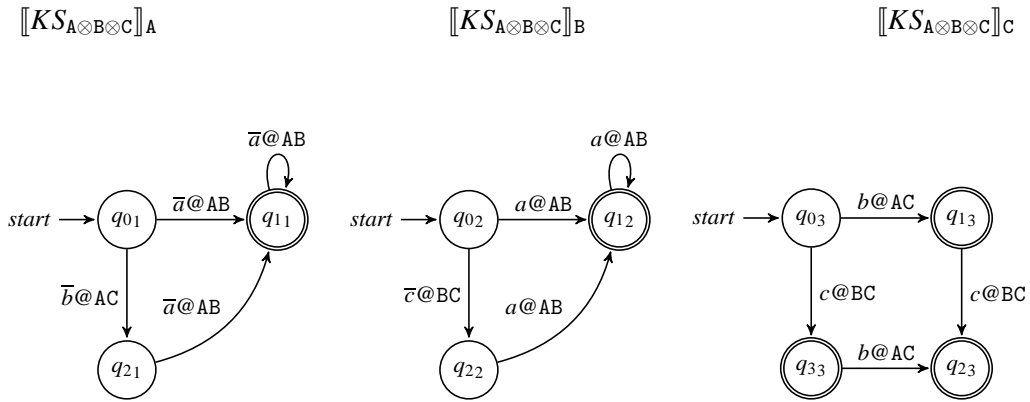
\[
\sem{KS_{\ptp A \otimes \ptp B \otimes \ptp C}}{\ptp A}   \qquad \qquad \qquad \qquad \qquad   
\sem{KS_{\ptp A \otimes \ptp B \otimes \ptp C}}{\ptp B}    \qquad \qquad \qquad \qquad \qquad   
\sem{KS_{\ptp A \otimes \ptp B \otimes \ptp C}}{\ptp C}
\]

\[
\begin{tikzpicture}[->,>=stealth',shorten >=1pt,auto,node distance=2.5cm,
                    semithick, every node/.style={scale=0.8}]
  \tikzstyle{every state}=[fill=white,draw=black,text=black]

  \node[initial,state] (A)                   	 {${q_0}_1$};
  
  \node[state, accepting] (B)  [right of=A]              	{${q_1}_1$};
  
  \node[state] (C)	[below of=A]         {${q_2}_1$};

  \path (A)			edge              node{$\PSEND{AB}{a}$} (B)
					edge              node[left]{$\PSEND{AC}{b}$} (C)
        (C)			edge	[bend right]			  node{$\PSEND{AB}{a}$}(B)
        	(B)			edge[loop above]	  node{$\PSEND{AB}{a}$}(B);				
					
\end{tikzpicture} \quad
\begin{tikzpicture}[->,>=stealth',shorten >=1pt,auto,node distance=2.5cm,
                    semithick, every node/.style={scale=0.8}]
  \tikzstyle{every state}=[fill=white,draw=black,text=black]

  \node[initial,state] (A)                   	 {${q_0}_2$};
  
  \node[state, accepting] (B)  [right of=A]              	{${q_1}_2$};
  
  \node[state] (C)	[below of=A]         {${q_2}_2$};

  \path (A)			edge              node{$\PRECEIVE{AB}{a}$} (B)
					edge              node[left]{$\PSEND{BC}{c}$} (C)
        (C)			edge	[bend right]			  node{$\PRECEIVE{AB}{a}$}(B)
        	(B)			edge[loop above]	  node{$\PRECEIVE{AB}{a}$}(B);			
					
\end{tikzpicture} \quad
\begin{tikzpicture}[->,>=stealth',shorten >=1pt,auto,node distance=2.5cm,
                    semithick, every node/.style={scale=0.8}]
  \tikzstyle{every state}=[fill=white,draw=black,text=black]

  \node[initial,state] (A)                   	 {${q_0}_3$};
  
  \node[state,accepting] (B)  [right of=A]        {${q_1}_3$};
  
  \node[state,accepting] (C)	[below of=B]         {${q_2}_3$};
   
  \node[state,accepting] (D)	[below of=A]         {${q_3}_3$};

  \path (A)			edge              node{$\PRECEIVE{AC}{b}$} (B)
  					edge				  node{$\PRECEIVE{BC}{c}$} (D)
        (B)			edge				  node{$\PRECEIVE{BC}{c}$}(C)
        	(D)			edge				  node{$\PRECEIVE{AC}{b}$}(C);				
					
\end{tikzpicture}
\]
\end{example}
We introduce the \emph{1-buffer} semantics of communicating machines, recall
that $\vec{w}$ is the vector of buffers.
Intuitively, this semantics forbids a machine to send a message to one
of its partners if there is a non empty channel in the system.

\begin{definition}[1-buffer, deadlock,convergent]\label{def:1b}
  A configuration $(\vec{q};\vec{w})$ of a CS $S =
  (M_\p)_{\p\in\PSet}$ is \emph{stable} if and only if
  $\vec{w} = \vec \emptyword$, while is \emph{final} if it
  is stable and $\vec q \in (F_\p)_{\p \in \PSet}$.
  The \emph{1-buffer semantics of $S$} is given by the relation
  \[
  \twoheadrightarrow \ \mmdef \ \rightarrow \cap (\RS_{\leq 1}(S) \times
  \Act \times \RS_{\leq 1}(S))
  \]
  where  $\rightarrow$ is the relation
  introduced in Definition~\ref{def:rs} and  
  \[ \RS_{\leq 1}(S)
  \mmdef \{(\vec{q};\vec{w}) \in \RS(S) \st (\vec{q};\vec{w})
  \text{ is stable or }
  \exists \p \q \in \chanset \qst \exists a \in \Rset \qst
  \ithel{\vec{w}}{\p\q}=a 
  \land
  \forall \ptp{p'}\ptp{q'} \neq \p\q. \ithel{\vec{w}}{\ptp{p'}\ptp{q'}} = \epsilon\}
  \]
  We say that the system $S$ is \emph{convergent} if and only if for every \emph{reachable} configuration
   $(\vec q;\vec w)$ it is possible to reach a final configuration in the 1-buffer semantics.
  
  Moreover a configuration $(\vec{q};\vec{w})$ is a \emph{deadlock} if
  and only if is not final and
  $(\vec{q};\vec{w}) \not \twoheadrightarrow$

\end{definition}

Note that if a system is \emph{convergent} than it is \emph{deadlock-free}.
The 1-buffer semantics above is instrumental to the relation we
establish between strong agreement of contract automata and
convergence of CFMSs.

\begin{remark}
Note that by considering only finite traces, we rule out all the unfair traces.
For example, consider the following \emph{strongly safe} CA:
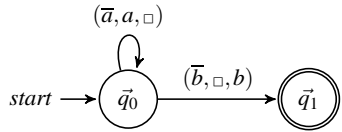
\[
\begin{tikzpicture}[->,>=stealth',shorten >=1pt,auto,node distance=3cm,
                    semithick, every node/.style={scale=0.8}]
  \tikzstyle{every state}=[fill=white,draw=black,text=black]

  \node[initial,state] (A)                   	 {$\vec{q_0}$};
  
  \node[state,accepting] (B)  [right of=A]              {$\vec{q_1}$};

  \path (A)			edge[loop above]            node{$(\overline{a},a,\blk)$} (A)
					edge             node{$(\overline b,\blk,b)$} (B);
\end{tikzpicture}
\]

If the first and second participants could execute the transition
$(\vec q_0, (\overline a,a,\blk),\vec q_0)$ infinitely often 
then the third participant would be prevented from reading the message $b$.
This behaviour is ruled out by considering only finite traces.
Indeed all the possible traces generated by the automaton above are described by the regular expression  
$(\overline a,a,\blk)^*(\overline b,\blk,b)$, where the third participant
will eventually reach its goal.
\end{remark}
We define $snd(\vec a) \mmdef \ptp i$ when $\vec a$ is a match action
or an offer action such that $\ithel{\vec{a}}{ i} \in \Oset$ and,
similarly, $rcv(\vec a) \mmdef \ptp j$ when $\vec a$ is a request or a match and
$j$ is such that $\ithel{\vec{a}}{j} \in \Rset$.

\begin{property}\label{pro:tr}
Let $S(KS_{{\aca A}})$ be a CS obtained by
Definition \ref{def:tr}, and $s_0$ be its initial configuration.
Then for all $f$ such that 
$s_0 \stackrel{f}{\twoheadrightarrow} $ there is a strong agreement $\varphi$
such that $f=\llbracket \varphi \rrbracket$ or 
$f=\llbracket \varphi \rrbracket \PSEND{ij}{a} \ $ for some $a,\ptp i,\ptp j$.
\end{property}
\begin{proof}
The proof follows trivially by observing that $KS_{{\aca A}}$ contains
only match transitions and that the 1-buffer semantics does not 
allow other behaviours for the CS.
\qed
\end{proof}

Before providing the main results, we introduce a notion of well-formedness of contract automata. 
We require that an output action of a participant in a particular state is independent from
the states of the other participants in the system.

\begin{definition}[Branching Condition]
A contract automaton $\aca A$ has the \emph{branching condition} iff for each
$\vec q_1,\vec q_2$ reachable in $\aca A$
the following holds

\[
\forall \text{$\vec{a}$ match actions } . (\vec{q}_1 \TRANSS{\vec{a}} \wedge snd(\vec{a})=\ptp i \wedge \ithel{\vec{q_1}}{i}=\ithel{\vec{q_2}}{i}) \text{ implies }  \vec{q}_2 \TRANSS{\vec{a}}
\]

\end{definition}
\begin{example}
Consider the CAs of Example \ref{ex:translation}. The product automaton has the branching
condition while this is not true for the MPC. Indeed, we have
$\ithel{\vec{q_0}}{1}=\ithel{\vec{q_3}}{1}={q_0}_1$ and 
$\vec{q_3} \TRANSS{(\overline a,a,\blk)}$ while there is no $\vec q_i$ such that $(\vec{q_1},  (\overline a, a,\blk), \vec q_i) \in T_{KS}$. 
\end{example}

The next theorem characterizes the relations between a CA 
and the corresponding CS. It states that the CS is capable 
of performing all the moves of the controller of the  CA,
while the CA is capable of performing those traces 
of the CS that are in \emph{strong agreement}. 
Moreover, if a CA has the branching condition, the runs of the 
CS leading
to a deadlock configuration correspond in CA to runs traversing
liable transitions, 
and likewise when the CS reaches
a non-deadlock configuration which does not
reach a final configuration.

\begin{theorem}\label{lem:move}
  Let $S(KS_{{\aca A}})$ be the communicating system obtained by the
  MPC $KS_{{\aca A}}$ with $s_0$, $s_1$, and $s_2$ be the initial
  configurations of $\aca A$, $KS_{{\aca A}}$, and $S(KS_{{\aca A}})$,
  respectively.
  The following hold:
  \begin{enumerate} 
  \item \label{move1}if $s_1 \TRANSS{\varphi} $ then $s_2 \stackrel{\llbracket
      \varphi \rrbracket}{\twoheadrightarrow}$
  \item \label{move2} \label{it:complete} if $s_2 \stackrel{\llbracket \varphi
      \rrbracket}{\twoheadrightarrow}$ and $\varphi$ is a strong agreement then
    $s_0 \TRANSS{\varphi}$
  \item \label{move3} if $s_2 \stackrel{f}{\twoheadrightarrow}$ reaches a deadlock configuration where $f=\llbracket \varphi
      \rrbracket$  or $f=\llbracket \varphi
      \rrbracket\PSEND{ij}{a}$ and the branching condition holds in $\aca A$
    then  $s_0 \TRANSS{\hat{\varphi}}$ has traversed a transition in $\tliable{KS_{{\aca A}}/{\aca A}}$ where $\hat{\varphi}$ can be respectively $\hat{\varphi}=\varphi$ or $\hat{\varphi} = \varphi\vec{a}$ where $\vec{a}$ is a match on $a$ with $snd(\vec{a})=\ptp i, rcv(\vec{a})= \ptp j$.
    \item \label{move4} if $s_2 \stackrel{\llbracket \varphi \rrbracket}{\twoheadrightarrow} s_2'$, $s_2'$ is not a deadlock configuration and no final configurations are reachable from $s_2'$     then  $s_0 \TRANSS{\varphi}$ has traversed a transition in $\tliable{KS_{{\aca A}}/{\aca A}}$.
 %
 %
%  \item \label{move4} if there is $\varphi \in \mathfrak Z$ such that
%    $s_0 \TRANSS{\varphi}\TRANSS{\vec a}$ where $\vec a$ is a match or offer action
%    on $a$ and the last transition is in $\tliable{KS_{{\aca A}}/{\aca A}}$,
%     then $s_2 \TRANSOB{\sem \varphi {}}
%    \TRANSOB{\PSEND{\ptp{ix}} a}{}$ will eventually reach a deadlock configuration
%    (where $\ptp x$ stands for a participant or $-$ and $i = snd(\vec a)$).
  \end{enumerate}
\end{theorem}
\begin{proof}
  Through the proof assume that $s_0'$, $s_1'$ and $s_2'$ are such
  that $s_0 \TRANSS{\varphi} s_0'$, $s_1 \TRANSS{\varphi} s_1'$ and
  $s_2 \TRANSOB{\sem \varphi{}} s_2'$.
  \begin{enumerate}
  \item By induction on $\varphi$.
    Assume that $s_1 \TRANSS{\vec{a}} s'_1$ with $\vec{a}$ match
    action on $a$ where principal  $i$ makes the offer and
    principal $j$ makes the request on $a$.
    Let $\ithel{\vec{q_0}}{i}$ and $\ithel{\vec{q_0}}{j}$ be the initial
    states of participants $\ptp i$ and $\ptp j$ in $S(KS_{{\aca A}})$.
    By Definition \ref{def:tr}, we have  for some $\vec{q_1}$ and $\vec{q_2}$
    that
    \[ (\ithel{\vec{q_0}}{i},\PSEND{ij}{a},\ithel{\vec{q_1}}{i})
    \qquad \text{and} \qquad
    (\ithel{\vec{q_0}}{j},\PRECEIVE{ij}{a},\ithel{\vec{q_2}}{i})
    \]
    are transitions of participants $\ptp i$ and $\ptp j$, respectively.
    We have $s_2 \TRANSOB{\PSEND{ij}{a}} \ \TRANSOB{\PRECEIVE{ij}{a}} s_2'$
    since after the first transition participant $\ptp j$ remains
    in its initial state.

    When $|\varphi| > 1$, we have for a configuration $s''_1$ that $s_1
    \TRANSS{\varphi} s''_1 \TRANSS{\vec{a_1}} s'_1$.
    Hence $s_2 \stackrel{\sem
      \varphi {}}{\twoheadrightarrow} s''_2$ (by the induction hypothesis)
    and,
    since $\vec{a_1}$ is a match, with the same reasoning we can
    conclude $s''_2 \TRANSOB{\llbracket \vec{a_1} \rrbracket} s'_2$.

  \item The proof is again by induction.
    Assume $s_2  \TRANSOB{\sem{\vec a} {}} s_2'$, where $\vec{a}$ is a
    match on $a$ which involves (the principals $i$ and $j$
    corresponding to) participants $\ptp i$ and $\ptp j$.
    As before let $i$ perform the offer and $j$ the request.
    By Definition~\ref{def:prod} (of product), we have that there is a
    transition $(\vec{q_0},\vec{a},\vec{q})$ in $\aca A$ from its
    initial state.
    
    When $|\varphi| > 1$, we have $s_2 \TRANSOB{\sem \varphi {}} s''_2
    \TRANSOB{\sem {\vec a} {}} s_2'$ and there is $w' \in (\Lset^n)^*$
    such that (by the induction hypothesis) $s_0 \TRANSS{\varphi} s_0' = (\vec{q'},w')$ is a run in
    $\aca A$.
    Reasoning as in the base case, we conclude that $\aca A$ has a
    transition of the form $(\vec{q'},\vec{a},\vec{q''})$.

  \item Let $s_2'= \csconf q w$ be the deadlock configuration reached from $s_2$
    with $f$.
   \newpage
    We distinguish two cases:
    \begin{itemize}
    \item if $\vec{w} \neq \vec \emptyword$ (namely some buffer in
      $\vec{w}$ is not empty) then $\hat{\varphi}$ is not a strong
      agreement (in fact, if it were a strong agreement then the 1-buffer
      semantics of $S(KS_{\aca A})$ would yield $\vec w = \vec
      \emptyword$).
      Then by Property \ref{pro:tr} we have $f= \sem{\varphi}{}\PSEND
        {\ptp{ij}} a$ and $\hat{\varphi}=\varphi\vec{a}$.
        Moreover, Theorem~\ref{lem:move}.\ref{move1} guarantees that  a run $s_0
        \TRANSS{\varphi} s_0' = (\vec q',w')$ exists in $\aca A$
        and, by Definition~\ref{def:tr} and the branching condition, there is a transition $(\vec q',
        \vec a, \vec q'')$ in $\aca A$ where $\vec a$ is a match action on $a$. 
        By contradiction, assume that there is no liable transition in the run
        $s_0 \TRANSS{\varphi\vec a} (\vec q'',w'')$ in $\aca A$.
        Then, by construction (cf. Definition~\ref{def:controller}), the
        MPC of $\aca A$ has the same run, namely $s_1 \TRANSS{\varphi\vec
          a} (\vec q'',w'')$.
        Finally, by Theorem~\ref{lem:move}.\ref{move1}, $s_2'
        \TRANSOB{\PRECEIVE{\ptp{ij}} a}$, contradicting the hypothesis that $s_2'$ is a
        deadlock configuration. 
      
    \item if $\vec{w} = \vec \emptyword$ (namely, all buffers are
      empty) then, by definition of deadlock configuration of CFSMs
      (cf. Definition~\ref{def:1b}), the state $\vec{q}$ of
      configuration $s_2' = \csconf q w$ is not final, there is no
      participant ready to fire an output, and there is a participant
      waiting for an input on one of its buffers.
      The latter condition is guaranteed by the construction of
      CFSMs from controllers.
      Since $\vec w = \vec \emptyword$, we have $f= \sem{\varphi}{}$ and $\hat{\varphi}=\varphi$ where $\varphi \in \mathfrak Z$ and
      there is a run $\s_0 \TRANSS{\varphi} s_0' = (\vec{q_1},w')$ in
      $\aca A$ (by Theorem~\ref{lem:move}.\ref{move2}). Note that $\vec q_1$ is redundant in $KS_{\aca A}$ (otherwise by Theorem~\ref{lem:move}.\ref{move1},  $s_2'$ would not be a deadlock
      configuration)  and, by construction (Lemma~\ref{lem:controller}), $\vec
      q_1$ is removed from $KS_{\aca A}$.
      This implies that a liable transition has been traversed in
      $s_0 \TRANSS{\varphi} s_0'$. 
    \end{itemize}
 %%%
 %%%
 \item Wlog we can assume that $\varphi$ is a strong agreement ( otherwise we  
       have
        $s_2 \stackrel{\llbracket \varphi \rrbracket \PSEND{ij}{a}}{\twoheadrightarrow}s_2'$
 	    for some $\ptp i,\ptp j,\overline a$
 	   and since $s_2'$ is not a deadlock it is possible to perform 
 	   the step $s_2'\stackrel{\PRECEIVE{ij}{a}}{\twoheadrightarrow}s_2''$ and we have
 	   that $ \varphi  \vec{a}$ is a strong agreement where $\vec{a}$ is a match action with
 	   $\ithel{\vec a}{i}=\overline a, \ithel{\vec a}{j}=a$).   
 	   Moreover, from $s_2''$
 	   is not possible to reach a final state, and we apply the following
 	   reasoning to $s_2'', \varphi \vec a$ instead of $s_2', \varphi$.
 	   
 	   Assume by contradiction that $s_0 \TRANSS{\varphi} s_0'$ has traversed no 
 	   liable transitions. Then by Definition \ref{def:culpability} there exists 
 	   $\varphi'$ such that $s_0' \TRANSS{\varphi'} s_0''$ and $s_0''$ is a 
 	   final configuration. By Lemma $\ref{lem:controller}$ we must have
 	   $s_1 \TRANSS{\varphi\varphi'} s_1'$ where $s_1'$ is a final configuration. Hence
 	   by applying Theorem \ref{lem:move}.\ref{move1} we have 
 	   $s_2' \stackrel{\llbracket \varphi' \rrbracket }{\twoheadrightarrow}s_2'''$ where $s_2'''$ is a final configuration, obtaining a 
 	   contradiction.
    \qed
  \end{enumerate}
\end{proof}

% Se K muove allora [K] muove  
% (per def. di [])
% 
% ------------------------------------------
% 
% Se t e' una mossa liable allora [K] da quello stato va in deadlock
% 
% dim.
% 
% per assurdo assumo che ho
% 
% q-(A-$>$B:a)=t-$>$q'   Liable ma
% 
% [q] -AB!a-$>$ x -AB?a-$>$ xx   
% 
% cioe' ho attraversato una mossa liable senza andare in deadlock. 
% 
% Dato che t non e' in K, per def. di traduzione da K a CM devo avere
% 
% q -f-$>$$>$ q1-t-$>$   dove f e' A,B free, altrimenti A e B non sono ready in [q] a scambiare il messaggio.
% 
% 
% Ma questo significa anche che per def. di CA ho q-t-$>$q'-f-$>$$>$q1  quindi t non e' liable.
% 
% -------------------------------------------------
% 
% Se [K] fa due mosse di sync  AB! AB?  allora K muove un match
% 
% dim.
% per assurdo K non ha quella mossa, visto che le macchine sono ready per def. di CA il match e' una 
% mossa liable, ma allora [K] non puo' fare synch per il ris. precedente, assurdo.
% 
% ----
% -----------------------------------------------------------------------------------

Note that the converse of Theorem \ref{lem:move}.\ref{move3} does not hold. Indeed if a CA $\mathcal A$
passes through a liable transition, it can be that $ S (KS_{\mathcal A})$ never reaches a deadlock configuration.
\begin{example}
Consider the CAs and CMs of Example \ref{ex:translation}.  A possible trace belonging to
the system $S(KS_{\ptp A \otimes \ptp B \otimes \ptp C})$ is generated by the transitions:
$({q_0}_1,\PSEND{AB}{a},{q_1}_1),({q_0}_1,\PRECEIVE{AB}{a},{q_2}_1)$. By using Theorem \ref{lem:move}.\ref{move2} this trace corresponds to the liable transition $(\vec{q_0},(\overline a,a,\blk),\vec q_7)$ of the product automaton.

However after this two steps the system $S(KS_{\ptp A \otimes \ptp B \otimes \ptp C})$ will never reach a deadlock configuration. Indeed, it is always possible to perform the transitions:$({q_1}_1,\PSEND{AB}{a},{q_1}_1),({q_2}_1,\PRECEIVE{AB}{a},{q_2}_1)$.
Note that $S(KS_{\ptp A \otimes \ptp B \otimes \ptp C})$ is deadlock-free but not convergent.
\end{example}

%\begin{example}
%Consider the following CA $\mathcal A$:
%\[
%\begin{tikzpicture}[->,>=stealth',shorten >=1pt,auto,node distance=2.5cm,
%                    semithick, every node/.style={scale=0.8}]
%  \tikzstyle{every state}=[fill=white,draw=black,text=black]
%
%  \node[initial,state] (A)                   	 {$\vec{q_0}$};
%  
%  \node[state] (B)  [right of=A]              	{$\vec{q_1}$};
%  
%  \node[state,accepting] (C)	[right of=B]         {$\vec{q_2}$};
%  
%  \node[state,accepting] (D)  [below of=B]              {$\vec{q_3}$};
%
%  \path (A)			edge              node{$(\overline{a},a)$} (B)
%					edge              node{$(\overline b,b)$} (D)
%        (B)			edge[loop above]	  node{$(\overline{a},a)$}(B)
%        				edge				  node{$(\blk,b)$}(C)
%        	(C)			edge[loop above]	  node{$(\overline{a},\blk)$}(C);				
%					
%\end{tikzpicture}
%\]
%The only trace in strong agreement is $(\overline b,b)$. 
%The liable transition is $(\vec q_0, (\overline a,a),\vec q_1)$.
%Consider the corresponding trace of $ S ( \mathcal A)$ 
%that traverses the liable transition: $\PSEND{AB}{a}.\PRECEIVE{AB}{a}$. After this two steps, the CS will never reach a deadlock
% configuration since from state $\vec q_1$ it is always possible to perform the steps $\PSEND{AB}{a}.\PRECEIVE{AB}{a}$.
%\end{example}

We are now ready to state our main result: the controller of a CA has the branching condition if and only if the corresponding CS is convergent.
\begin{theorem}\label{the:deadlockfree}
  Let $\aca A$  be a contract automaton, $KS_{{\aca A}}$ be its MPC, and $S(KS_{\aca A})$ be the communicating system obtained by $KS_{\aca A}$.
 The following statements are equivalent :
  \begin{itemize}
  \item[1]  $S(KS_\mathcal{A})$ is convergent
  \item[2]  $KS_{\aca A}$ has the branching condition
 \end{itemize} 
\end{theorem}
\begin{proof}
 Let $s_0$, $s_1$, and $s_2$ be the initial
  configurations of $\aca A$, $KS_{{\aca A}}$, and $S(KS_{{\aca A}})$,
  respectively.
  
($\leftarrow$) Assume by contradiction that $S( KS_\mathcal{A})$ is not
  convergent and $KS_\mathcal{A}$  has the branching condition holds, namely there exists $s_2'=(\vec{q},\vec{w})$ such that $s_2 \stackrel{\varphi}{\twoheadrightarrow} s_2'$ and    
 no final configurations are reachable from  
  $(\vec{q},\vec{w})$. 
 We distinguish two cases:
 \begin{itemize}
 	\item  $s_2'$ is not a deadlock configuration. Then by applying 
 	Theorem~\ref{lem:move}.\ref{move4} we have $\varphi=\varphi_1
  \vec{a}\varphi''$ for some $\varphi_1,
  \vec{a},\varphi''$ such that 
 $\vec{q_0} \TRANSS{\varphi_1}
  \vec{q_1}$ is a run of ${\aca A}$ and
  $(\vec{q_1},\vec{a},\vec{q_1'}) \in
  \tliable{KS_\mathcal{A}/\mathcal{A}}$. Note that $\vec{a}$ is a
  match action,  otherwise a participant in $S( KS_\mathcal{A})$ has fired 
    an action to the environment in $\varphi$. Since $S( KS_\mathcal{A})$ is derived from $ KS_\mathcal{A}$
 and all the transitions of  $ KS_\mathcal{A}$ are match, this is not possible.
    Assume $\vec{a}$ is an action on $a$ with $snd(\vec a)=\ptp i$, $rcv(\vec{a})=\ptp j$ for some $\ptp i,\ptp j  \in \PSet$.

 By hypothesis we know that $s_2 \stackrel{\llbracket \varphi_1
    \rrbracket}{\twoheadrightarrow} s_2'   \stackrel{ \PSEND{ij}{a} }{ \twoheadrightarrow}$, hence in the configuration $s_2'$ the participant 
$\ptp i$ is able to fire the action $\PSEND{ij}{a}$.  By Definition \ref{def:tr}, \ref{def:prod}  there must be a state $\vec{q_2}$ in $KS_{\aca A}$ such that $(\vec{q_2},\vec{a},\vec{q_3})$ is a transition in $KS_{\aca A}$  (recall that  $S( KS_\mathcal{A})$ is derived from $ KS_\mathcal{A}$) and $\ithel{\vec{q_2}}{i}=\ithel{\vec{q_1}}{i}$ (otherwise we would not have $ s_2'   \stackrel{ \PSEND{ij}{a} }{ \twoheadrightarrow}$), and since $(\vec{q_1},\vec{a},\vec{q_1'}) \in
  \tliable{KS_\mathcal{A}/\mathcal{A}}$ we conclude that the branching condition does not hold in $KS_\mathcal{A}$, obtaining a contradiction.
  \item all the possible configurations $s_2'$ are deadlock. Then it must be that 
        $s_2  \stackrel{\llbracket \varphi
    \rrbracket}{\twoheadrightarrow} s_2''   \stackrel{ \PSEND{ij}{a} }{ \twoheadrightarrow} s_2'$ for some $\ptp i,\ptp j, \overline a$,
    and from $s_2''$ it is possible to reach a final
    configuration, that is $s_2''\stackrel{\llbracket \varphi'
    \rrbracket}{\twoheadrightarrow} s_2''' $ where $s_2'''$ is final.  
 Note that it is not possible to have
    $s_2''   \stackrel{ \PRECEIVE{ij}{a} }{ \twoheadrightarrow} s_2'$ otherwise
    we would have that from $s_2''$ , which is not a deadlock,     
    is not possible to reach a final configuration, or that $s_2'$ is not a deadlock.
    
    By
    Theorem \ref{lem:move}. \ref{move2} we have $s_0 \TRANSS{\varphi \varphi'} s_0'$
    where $s_0'$ is final, hence by Lemma $\ref{lem:controller}$ it must be
    $s_1 \TRANSS{\varphi} s_1'=(\vec q_1,w)$. 
    As the previous case,  
    by Definition \ref{def:tr}, \ref{def:prod}  there must be a state $\vec{q_2}$ in $KS_{\aca A}$ such that $(\vec{q_2},\vec{a},\vec{q_3})$ is a transition in $KS_{\aca A}$ where $\vec a$ is a match  an action on $a$ with $snd(\vec a)=\ptp i$, $rcv(\vec{a})=\ptp j$ and $\ithel{\vec{q_2}}{i}=\ithel{\vec{q_1}}{i}$.
    Moreover since $s_2'$ is a deadlock, it must be that
    there is no transition $(\vec q_1, \vec a, \vec q_4)$ in $KS_\aca A$, 
    otherwise by Theorem \ref{lem:move} . \ref{move1} we have 
     $s_2'   \stackrel{ \PRECEIVE{ij}{a} }{ \twoheadrightarrow}$,
     obtaining a contradiction.
    Hence we have that the branching condition does not hold in $KS_\aca A$,
    since there is no transition $(\vec q_1, \vec a, \vec q_4)$ in $KS_\aca A$.
   
 \end{itemize}

%%%

($\rightarrow$) By contradiction assume that the branching condition 
does not hold in $KS_{\aca A}$. Hence we have two states $\vec{q_1},\vec{q_2}$
 in $KS_{\aca A}$ such that $\vec{q_0} \TRANSS{\varphi} \vec{q}_1 \TRANSS{\vec{a}},
 \vec{q_0} \TRANSS{\varphi'} \vec{q}_2 \not \hspace{-5pt} \TRANSS{\vec{a}}$ where 
 $\vec{a}$ is a match on $a$ with $snd(\vec a)=\ptp i$, $rcv(\vec{a})=\ptp j$ for some
 $\ptp i,\ptp j  \in \PSet$  and $\ithel{\vec{q_1}}{i}=\ithel{\vec{q_2}}{i}$.

By Theorem~\ref{lem:move}.\ref{move1} we have $s_2 \TRANSS{\sem{\varphi}{}} s_2' 
\TRANSS{\PSEND{ij}{a}}$ and $s_2 \TRANSS{\sem{\varphi'}{}} s_2''$.  By Definition \ref{def:tr} and  
\ref{def:prod} we know that the participant $\ptp i$ is in the same state in the configuration $s_2',s_2''$,
hence we have $ s_2'' \TRANSS{\PSEND{ij}{a}} s_2'''$ and from $s_2'''$ is is not possible to reach
a final configuration. Otherwise 
if $ s_2''' \TRANSS{\PRECEIVE{ij}{a}\sem{\varphi_2}{}} s_f$ where $s_f$ is final, then by Theorem~\ref{lem:move}.\ref{move2} we would have
 $\vec{q}_2 \TRANSS{\vec{a}\varphi_2} \vec q_f$ where $\vec q_f$ is a final state of the CA,  hence $\vec{q}_2 \TRANSS{\vec{a}}$ is not
liable and belongs to  $KS_{\aca A}$, obtaining a contradiction.

\qed
\end{proof}

A consequence of Theorem \ref{the:deadlockfree} is that a \emph{strongly safe} CA has the branching condition if and only if the corresponding CS is convergent.
\begin{corollary}
Let $\aca A$ be a contract automaton, then $\aca A$ is strongly safe and has branching condition if and only if $S(\aca A)$ is convergent.
\end{corollary}
\begin{proof}
The statement follows trivially by notice that if $\aca A$ is strongly safe then $\aca A= KS_\mathcal{A}$, hence $KS_{\aca A}$ has the branching condition and we can apply Theorem \ref{the:deadlockfree}.

%  Assume that $s_2$ is the initial configuration of  $S(\aca A)$ and
%  proceed by contradiction.
%  % 
%  If we reach a deadlock configuration $s_2 \stackrel{f}{\twoheadrightarrow} (\vec{q},\vec{w})$, by
%  Theorem~\ref{lem:move}.\ref{move3}  we have traversed a liable transition in $\aca
%  A$, hence $\aca A$ is not strongly safe.
%  % 
  \qed
\end{proof}

\begin{example}
Consider the following \emph{strongly safe} CA $\mathcal A=\ptp A \otimes \ptp B \otimes \ptp C \otimes \ptp D$:
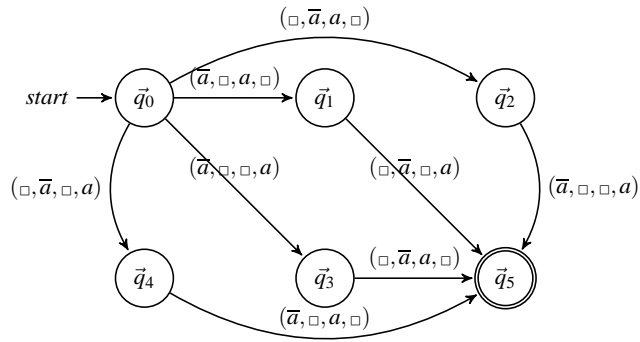
\[
\begin{tikzpicture}[->,>=stealth',shorten >=1pt,auto,node distance=3.0cm,
                    semithick, every node/.style={scale=0.8}]
  \tikzstyle{every state}=[fill=white,draw=black,text=black]

  \node[initial,state] (A)                   	 {$\vec{q_0}$};
  
  \node[state] (B)  [right of=A]              	{$\vec{q_1}$};
  
  \node[state] (C)	[right of=B]         {$\vec{q_2}$};
  
  \node[state] (D)  [below of=B]              {$\vec{q_3}$};
  
  \node[state]  (E) [below of=A]				{$\vec q_4$}; 
  
  \node[state,accepting]  (F) [below of=C]				{$\vec q_5$};

  \path (A)			edge              node{$(\overline{a},\blk,a,\blk)$} (B)
					edge              node[above]{$(\overline a,\blk,\blk,a)$} (D)
					edge[bend left]				  node{$(\blk,\overline{a},a,\blk)$} (C)
					edge	[bend right]			  node[left]{$(\blk,\overline{a},\blk,a)$}(E)
        (B)			edge				  node[above]{$(\blk,\overline{a},\blk,a)$}(F)
        (C)			edge	[bend left]			  node{$(\overline a,\blk,\blk,a)$}(F)
        	(D)			edge				  node{$(\blk,\overline{a},a,\blk)$}(F)
        	(E)			edge[bend right]	  node[above]{$(\overline{a},\blk,a,\blk)$}(F);				
\end{tikzpicture}
\]
In this example we have four participants: the first two ($\ptp A,\ptp B$) perform the
same offer $\overline a$, while the others ($\ptp C,\ptp D$) perform the request $a$. The CA $\aca A$ has no branching condition: for example the internal state of
the participant $\ptp B$ is the same in both states $\vec q_1,\vec q_3$.
From state $\vec q_1$ we have the match transition $(\blk,\overline a,\blk,a)$
which is not available in state $\vec q_3$, and from state $\vec q_3$ we have 
the match transition  $(\blk,\overline a,a,\blk)$ which is not available from
state $\vec q_1$.

The translation yields the CMs:
\[
\sem{KS_{\aca A}}{\ptp A}=\PSEND{AC}{a}+\PSEND{AD}{a}
\qquad
\sem{KS_{\aca A}}{\ptp B}=\PSEND{BC}{a}+\PSEND{BD}{a}
\]
\[
\sem{KS_{\aca A}}{\ptp C}=\PRECEIVE{AC}{a}+\PRECEIVE{BC}{a}
\qquad
\sem{KS_{\aca A}}{\ptp D}=\PRECEIVE{AD}{a}+\PRECEIVE{BD}{a}
\]
A deadlock configuration is generated
by the trace $\PSEND{AC}{a}.\PRECEIVE{AC}{a}.\PSEND{BC}{a}$.

\end{example}
\begin{figure}[tb]
  \center

  \begin{tikzpicture}[->,>=stealth',shorten >=1pt,auto,node distance=3cm,
    semithick, every node/.style={scale=0.7}]
    \tikzstyle{every state}=[fill=white,draw=black,text=black]

    \node[initial,state] (A)                   	 {$\vec{q_0}$};
    
    \node[state] (B)  [right of=A]                  			 {$\vec{q_1}$};

    \node[state] (C)  [right of=B]                 			 {$\vec{q_2}$};

    \node[state,accepting] (D)  [right of=C]                 	 {$\vec{q_3}$};

    \node[state] (E)  [below of=A]{$\vec{q_4}$};

    \node[state] (F)  [right of=E]{$\vec{q_5}$};

    \node[state] (G)  [right of=F]{$\vec{q_6}$};

    \node[state] (H)  [right of=G]{$\vec{q_7}$};

    \node[state] (I)  [below of=G]{$\vec{q_8}$};

    \node[state] (L)  [below of=H]{$\vec{q_9}$};

    \node[state] (Z)  [right of=L,color=red,text=white]{$\bot$};

    \path (A)			edge            node{$(\aout,\blk,\areq)$} (B)
    edge[left]            node{$(\aout,\areq,\blk)$} (E)
    edge		   node{$(\blk,\creq,\cout)$} (F);

    \path (B)			edge            node{$(\okreq,\blk,\okout)$} (C);

    \path (C)			edge            node{$(\dreq,\blk,\dout)$} (D);

    \path (E)			edge [bend right]           node{$(\blk,\creq,\cout)$} (I);

    \path (F)			edge            node{$\hspace{-25pt}(\aout,\areq,\blk)$} (I);

    \path (I)			edge [right]           node{$(\okreq,\okout,\blk)$} (G)
    edge            node{$(\blk,\okout,\okreq)$} (L);

    \path (G)			edge            node{$(\blk,\okout,\okreq)$} (H);

    \path (H)			edge            node{$(\dreq,\blk,\dout)$} (D);

    \path (L)			edge[right]            node{$(\okreq,\okout,\blk)$} (H)
    edge[color=red]            node{$(\blk,\blk,\dout)$} (Z);

  \end{tikzpicture}
  \caption{$KS_{\aca A}$}
\label{fig:badexample}
\end{figure}
\begin{example}

  Figure~\ref{fig:badexample} depicts the automaton $KS_{\aca A}$ 
  where $\aca A=\ptp A \otimes \ptp B \otimes \ptp C$:
\[
\ptp A=\overline{a}.ok.d \qquad
\ptp B=(a.c + c.a).\overline{ok}.\overline{ok} \qquad
\ptp C=a.\overline{ok}.\overline{d} + \overline{c}.ok.\overline{d}
\]
Participant $\ptp A$ sends an offer $\overline{a}$ and then waits 
on acknowledgement $ok$ and  then a message $d$.
Participant $\ptp B$ acts as an intermediary: it receives the requests $a$
and $c$ and then replies with $\overline{ok}$.
Finally, participant $\ptp C$ can either behave similarly to $\ptp A$
or directly acknowledge the
message received on $a$ (and then send~$d$).
The translation in Definition~\ref{def:tr} yields the following
communicating machines, written as regular expressions:

\[
\sem{KS_{\aca A}}{\ptp A}=
\PSEND{AC}{a}.\PRECEIVE{CA}{ok}.\PRECEIVE{CA}{d}+ 
\PSEND{AB}{a}.\PRECEIVE{BA}{ok}.\PRECEIVE{CA}{d} \]
\[
\sem{KS_{\aca A}}{\ptp B}=(\PRECEIVE{AB}{a}.\PRECEIVE{CB}{c} + \PRECEIVE{CB}{c}.\PRECEIVE{AB}{a}).(\PSEND{BA}{ok}.\PSEND{BC}{ok} + \PSEND{BC}{ok}.\PSEND{BA}{ok})
\]
\[
\sem{KS_{\aca A}}{\ptp C} = \PSEND{CB}{c}.\PRECEIVE{BC}{ok}.\PSEND{CA}{d} + \PRECEIVE{AC}{a}.\PSEND{CA}{ok}.\PSEND{CA}{d}
\]

Note that  $KS_{\aca A}$ has no 
branching condition, indeed $\ithel{\vec{q_9}}{3}=\ithel{\vec{q_7}}{3}$  
but  there is no $(\vec{q_9},(d,\blk,\overline{d}),\vec{q'})$ in $KS_{\aca A}$
for some $\vec{q'}$. Moreover there is a \emph{liable}
transition with label $(\vec{q_9},(\blk,\blk,\overline{d}),\bottom)$ in
Figure~\ref{fig:badexample} which represents the possible
deadlock in the system.

A deadlock configuration in $S(KS_{\aca A})$ is given by the trace $\llbracket \varphi \rrbracket \PSEND{CA}{d}$ where:
\[
\varphi=(\overline{a},a,\blk)(\blk,c,\overline{c})(,\overline{ok},ok)
\]

Indeed $s_0 \TRANSOB{\llbracket \varphi \rrbracket} \TRANSOB{ \PSEND{CA}{d}} s_0'$ where $s_0$ is the initial configuration of $S(KS_{\aca A})$. In the configuration $s_0'=(\vec{q},\vec{w})$ the buffer $\vec{w}$ is not empty, because $\ithel{\vec{w}}{CA}=d$. Moreover the machine $A$ is prevented to read the message on the buffer since its configuration in $\vec{q}$ is $\PRECEIVE{BA}{ok}.\PRECEIVE{CA}{d}$. 
%Note that by refining the machine $\ptp B$ to:
%
%\[
%B'=(\PRECEIVE{AB}{a}.\PRECEIVE{CB}{c} + \PRECEIVE{CB}{c}.\PRECEIVE{AB}{a}).(\PSEND{BA}{ok}.\PSEND{BC}{ok})
%\]
%
%\noindent
%we would obtain a deadlock free system. We leave such a fine grained analysis as a future work.
\end{example}

%%% Local Variables: 
%%% mode: latex
%%% TeX-master: "main"
%%% End: 

\section{On extending the approach}\label{sec:extensions}
	% !TEX root = main.tex
% 

We discuss possible extensions of
our approach to other existing types of 
agreement on CAs, and on different
 semantics of CMs, where there are no constraints on
 the number of messages in a buffer.
 We start by comparing the other existing types of agreement with
 the 1-buffer semantics for CMs.
\paragraph{On agreement}  The property of \emph{agreement} requires that all the requests
are matched. It allows strings made by match and offer actions only. In the following we discuss a correspondence similar to Theorem~\ref{the:deadlockfree} 
for  the property of \emph{agreement}. 
\begin{example}\label{ex:agreement}
Consider the CAs corresponding to the regular expressions $\ptp A=\overline b.d + \overline c.e + d.e$ and 
$\ptp B=\overline d.\overline e$. The controller
$\mathcal K_{A \otimes B}$ for the property of \emph{agreement} is given
by the CA: 
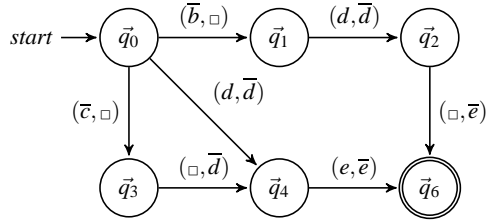
\[
\begin{tikzpicture}[->,>=stealth',shorten >=1pt,auto,node distance=2.5cm,
                    semithick, every node/.style={scale=0.8}]
  \tikzstyle{every state}=[fill=white,draw=black,text=black]

  \node[initial,state] (A)                   	 {$\vec{q_0}$};
  
  \node[state] (B)  [right of=A]              	{$\vec{q_1}$};
  
  \node[state] (C)	[right of=B]         {$\vec{q_2}$};
  
  \node[state] (D)  [below of=A]              {$\vec{q_3}$};
  
  \node[state] (E)	[right of=D]				  {$\vec q_4$};
  
  \node[state,accepting] (G)	[right of=E]				  {$\vec q_6$};

  \path (A)			edge              node{$(\overline{b},\blk)$} (B)
					edge              node[left]{$(\overline c,\blk)$} (D)
					edge				  node{$(d,\overline{d})$} (E)
        (B)			edge				  node{$(d,\overline d)$}(C)
        (C)			edge				  node{$(\blk,\overline e)$}(G)
        (D)			edge				  node{$(\blk, \overline d)$}(E)
        (E)			edge				  node{$(e,\overline e)$}(G);			
\end{tikzpicture}
\]

The translation in Definition \ref{def:tr} yields the CMs:

%\[
%\sem{\mathcal K_{A \otimes B}}{\ptp A}=\PSEND{A-}{b}.\PRECEIVE{BA}{d}
%+ \PSEND{A-}{c}.\PRECEIVE{BA}{e}+\PRECEIVE{BA}{d}.\PRECEIVE{BA}{e}
%\qquad
%\sem{\mathcal K_{A \otimes B}}{\ptp B}= \PSEND{BA}{d}.\PSEND{B-}{e}
%+ \PSEND{B-}{d}.\PSEND{BA}{e}+\PSEND{BA}{d}.\PSEND{BA}{e}
%\]
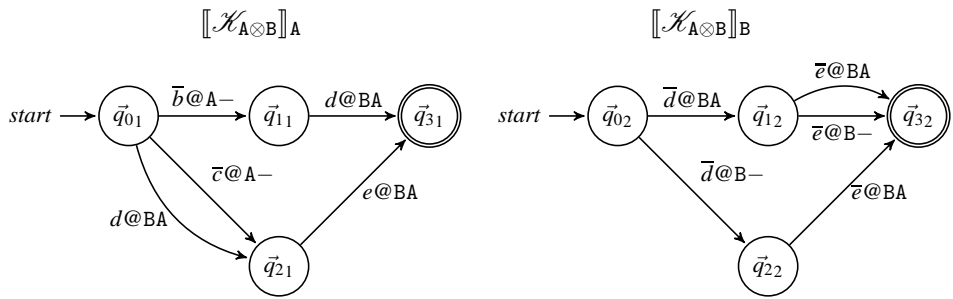
\[
\sem{\mathcal{K}_{\ptp A \otimes \ptp B}}{\ptp A}  
\qquad \qquad \qquad \qquad \qquad \qquad
\sem{\mathcal{K}_{\ptp A \otimes \ptp B}}{\ptp B}  
\]
\[
\begin{tikzpicture}[->,>=stealth',shorten >=1pt,auto,node distance=2.5cm,
                    semithick, every node/.style={scale=0.8}]
  \tikzstyle{every state}=[fill=white,draw=black,text=black]

  \node[initial,state] (A)                   	 {$\vec{q_0}_1$};
  
  \node[state] (B)  [right of=A]              	{$\vec{q_1}_1$};
  
  \node[state] (C)	[below of=B]         {$\vec{q_2}_1$};
  
 \node[state, accepting] (D)	[right of=B]         {$\vec{q_3}_1$};
  
  \path (A)			edge              node{$\PSEND{A-}{b}$} (B)
					edge              node{$\PSEND{A-}{c}$} (C) 
					edge[bend right]              node[left]{$\PRECEIVE{BA}{d}$} (C) 
  		(B)			edge              node{$\PRECEIVE{BA}{d}$} (D)
  	 	(C)			edge              node[right]{$\PRECEIVE{BA}{e}$} (D);		
\end{tikzpicture} \quad
\begin{tikzpicture}[->,>=stealth',shorten >=1pt,auto,node distance=2.5cm,
                    semithick, every node/.style={scale=0.8}]
  \tikzstyle{every state}=[fill=white,draw=black,text=black]

  \node[initial,state] (A)                   	 {$\vec{q_0}_2$};
  
  \node[state] (B)  [right of=A]              	{$\vec{q_1}_2$};
  
  \node[state] (C)	[below of=B]         {$\vec{q_2}_2$};
  
 \node[state, accepting] (D)	[right of=B]         {$\vec{q_3}_2$};
  
  \path (A)			edge              node{$\PSEND{BA}{d}$} (B)
					edge              node{$\PSEND{B-}{d}$} (C) 
					
  		(B)			edge              node[below]{$\PSEND{B-}{e}$} (D)
  		 	 		edge	[bend left]   node{$\PSEND{BA}{e}$}(D)
  	 	(C)			edge              node[right]{$\PSEND{BA}{e}$} (D);		
\end{tikzpicture}
\]
Under the 1-buffer semantics, the system $S(\mathcal K_{A \otimes B})$
always reaches a deadlock configuration
since in every execution there are messages in the buffer with no receiver, corresponding to the offer actions
in the controller. 

If we assume that the unmatched offers  are consumed instantaneously by
an artificial participant representing the environment, we still have possible 
deadlock configurations in $S(\mathcal K_{A \otimes B})$, for example if the participants execute the sequence of transitions 
$(\vec{q_0}_1, 	\PSEND{A-}{b},\vec{q_1}_1)$, 
$(\vec{q_0}_2, 	\PSEND{B-}{d},\vec{q_2}_2)$, 
$(\vec{q_2}_2, 	\PSEND{BA}{e},\vec{q_3}_2)$.

Note that $\mathcal K_{A \otimes B}$ has no branching condition: 
in $\vec q_1, \vec q_3$ the participant $\ptp B$ is the state $\vec{q_0}_2$, 
but from $\vec q_3$ there is no match transition on action $d$. 
\end{example}

Under the assumptions that the offer actions are consumed by the environment, that it is possible to prove that 
the controller of the CA $\mathcal A$ has a slightly modified version of the 
branching condition if and only if the corresponding
system $S(\mathcal K_{\mathcal A})$ is convergent. The proof is obtained by noticing that in the
1-buffer semantic a deadlock configuration is reached only if a participant
$\ptp A$ send a message $a$ to a participant $\ptp B$ and $\ptp B$ is unable to consume the message.
By Definition \ref{def:tr} this can happen only if there are two different states in the CA where $\ptp A$ is in
the same internal state and the match transition is available only in one of the two states, i.e. $\mathcal K_{\mathcal A}$ has no branching condition.
We also need to consider those configurations which are not convergent
nor deadlocks as done in Theorem 
\ref{the:deadlockfree}.

\paragraph{On weak agreement} For the property of \emph{weak agreement} things are more intricate, indeed it is necessary
to modify the actual translation. This is due to the possibility for a participant to fire a request
if in the future the offer will be available, while in the CMs if the buffer is empty it is 
not possible to perform an input action.

To overcome this problem it is possible to synthesize one or more CMs which act as brokers.
They receive as input all the actions of the participants, which are now translated into outputs, and reply with messages in a way to drive the participants 
through the trace in weak agreement.
\paragraph{On different semantics } We now discuss the relations between CAs and other semantics for CMs. 
\begin{example}
Consider the following CA $\ptp A \otimes \ptp B$:
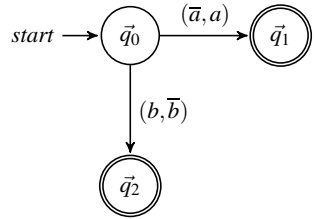
\[
\begin{tikzpicture}[->,>=stealth',shorten >=1pt,auto,node distance=2.5cm,
                    semithick, every node/.style={scale=0.8}]
  \tikzstyle{every state}=[fill=white,draw=black,text=black]

  \node[initial,state] (A)                   	 {$\vec{q_0}$};
  
  \node[state,accepting] (B)  [right of=A]              	{$\vec{q_1}$};
  
  \node[state, accepting] (C)	[below of=A]         {$\vec{q_2}$};

  \path (A)			edge              node{$(\overline{a},a)$} (B)
					edge              node{$(b, \overline b)$} (C);			
\end{tikzpicture}
\]
We have $\ptp A=\overline a + b$, $\ptp B=\overline b + a$. This CA is \emph{strongly safe} and has the branching condition. However, by considering 
the non 1-buffer semantics for CMs, the translated system is not $\emph{deadlock free}$.
Indeed a possible deadlock in $S(\mathcal K_{A \otimes B})$ is generated if the first participant performs the action $\PSEND{AB}{a}$ and then the second performs the action $\PSEND{BA}{b}$. 
This is because participant $\ptp B$ can ignore the message received by the 
participant $\ptp A$ and follow the other branch of the CA. These behaviours is not permitted
by the 1-buffer semantics, which forces participants to follow the successful branch.
\end{example}

The previous example shows that if we allow a less constrained semantics for CMs then 
Theorem \ref{the:deadlockfree} does not hold any more. Indeed, we need to introduce other
constraints on the behaviour of the CAs to obtain a correspondence with convergent
systems.

Note that in the previous example, from state $\vec q_0$ both participants contain
a branch where they can execute an input or an output action. This is called a
\emph{mixed choice} state.
It is possible to prove that if a CA has the branching condition, it is strongly safe and has no
mixed choice states then the
corresponding system is convergent with non 1-buffer semantics.
However the converse does not hold.
Indeed there exists systems with mixed choice states that enjoy convergence.

\begin{example}
Consider the following CA $\ptp A \otimes \ptp B$:
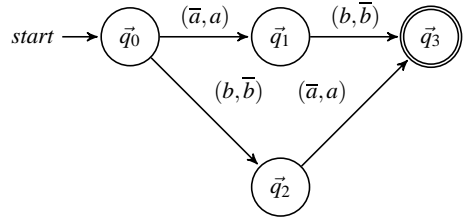
\[
\begin{tikzpicture}[->,>=stealth',shorten >=1pt,auto,node distance=2.5cm,
                    semithick, every node/.style={scale=0.8}]
  \tikzstyle{every state}=[fill=white,draw=black,text=black]

  \node[initial,state] (A)                   	 {$\vec{q_0}$};
  
  \node[state] (B)  [right of=A]              	{$\vec{q_1}$};
  
  \node[state] (C)	[below of=B]         {$\vec{q_2}$};
  
 \node[state, accepting] (D)	[right of=B]         {$\vec{q_3}$};

  \path (A)			edge              node{$(\overline{a},a)$} (B)
				edge              node{$(b, \overline b)$} (C) 
  	 (B)			edge              node{$(b, \overline b)$} (D)
  	 (C)			edge              node{$(\overline{a},a)$} (D);		
\end{tikzpicture}
\]
We have $\ptp A= \overline a.b +b.\overline a$, $\ptp B=\overline b.a + a.\overline b$.
This CA is \emph{strongly safe}, has the branching condition, and contains a mixed choice state, i.e. $\vec q_0$. Nevertheless, 
the corresponding system is convergent.
\end{example}

As showed by the previous example, for obtaining a correspondence similar to Theorem \ref{the:deadlockfree},
we need to consider those \emph{bad} mixed choices, where the participants behave differently 
in the different branches.
\section{Concluding Remarks}\label{sec:conc}

We have established a formal correspondence between contract automata, an
orchestration model, and communicating machines, a model of
choreography.
An interesting implication of our results is that contract automata can
be seen as an alternative semantics of communicating machines.
In fact, the product of communicating machines could be built
as a
contract automaton once match-actions are properly defined as
tuples where 
output messages appear before the corresponding input ones.
However, contract automata are more general in the sense that they
would also admit matches where a request appears before its
corresponding offer.
Exploring those alternative semantics is of interest and it is scope
for future work.

The dichotomy orchestration-choreography has been discussed in many
papers (see e.g.,~\cite{pel03}).
The only formal results (we are aware of) that link a choreography
to an orchestration framework is in~\cite{bglz06}.
A precise comparison with~\cite{bglz06} is not straightforward  as
the models use a bisimulation-like relation to exhibit a conformance
relation between choreographed and orchestrated computations.
Here, we study the conditions to ``force'' orchestrated computations
to well-behave (strong safety), and convergence in a choreography
framework in terms of strong safety in the orchestration one.

A practical outcome of our result is that strong safe contract
automata can execute without controller (if they are trusted).
In fact, one can translate them into communicating machines that run
without central control.

For the time being, our result only states that strong agreement corresponds to
the 1-buffer semantics of communicating machines.
In other words, the execution of the machines is basically synchronous.
(We note that this has some advantages since communicating machines
with 1-buffer semantics are more computationally tractable~\cite{cf05}.)
We conjecture that results similar to the one presented in this paper
can be achieved for weaker notions of agreement (for example, the ones
in~\cite{BasileDF14}) when considering asynchronous behaviours of
communicating machines.
This is nonetheless left as future work.

%%% Local Variables: 
%%% mode: plain-tex
%%% TeX-master: "main"
%%% End: 

\bibliographystyle{eptcs}
\bibliography{bib}

\end{document}